\documentclass[11pt]{article}
\usepackage{amsmath,amsthm}
\usepackage{amssymb,latexsym}
\usepackage{epsfig}
\usepackage{hyperref}
\usepackage{float}
\usepackage{mathpazo}
\usepackage{fullpage}
\usepackage{enumerate}

\newtheorem{theorem}{Theorem}
\newtheorem{corollary}[theorem]{Corollary}
\newtheorem{lemma}[theorem]{Lemma}

\newtheorem{definition}[theorem]{Definition}
\newtheorem{claim}[theorem]{Claim}
\newtheorem{fact}[theorem]{Fact}

\newcommand{\poly}{{\rm poly}}
\newcommand{\eps}{\varepsilon}
\newcommand{\bra}[1]{\langle #1|}
\newcommand{\ket}[1]{|#1\rangle}

\newcommand{\C}{\mathbb{C}}
\newcommand{\Tr}{Tr}
\newcommand{\Trho}{Tr_\rho} 
 
\newcommand{\Id}{Id} 

\newcommand{\Exs}[2]{E_{#1}\left[#2\right]} 

\newcommand{\beq}{\begin{equation}}
\newcommand{\eeq}{\end{equation}}

\usepackage{color,graphicx}

\begin{document}

\title{Parallel Repetition of Entangled Games}
\author{ Julia Kempe\footnote{Blavatnik School of Computer Science, Tel Aviv University, Tel Aviv 69978, Israel and CNRS \& LRI,
 University of Paris-Sud, Orsay, France. Supported
by an Individual Research Grant of the Israeli Science Foundation, by European Research Council (ERC) Starting Grant
QUCO and by the Wolfson Family Charitable Trust.}
 \and Thomas Vidick\footnote{Computer Science division, UC Berkeley, USA. Supported by ARO Grant W911NF-09-1-0440 and NSF Grant CCF-0905626. Part of this work while done while visiting LRI, University of Paris-Sud, Orsay, France.}}
\date{\today}
\maketitle

\begin{abstract}
We consider one-round games between a classical referee and two players. One of the main questions in this area is the
{\em parallel repetition question}: Is there a way to decrease the maximum winning probability of a game
without increasing the number of rounds or the number of players? Classically, efforts to resolve this question, open for many years, have
culminated in Raz's celebrated parallel repetition theorem on one hand, and in efficient product testers for PCPs on the other.

In the case where players share entanglement, the only previously known results are for special cases of games, and
are based on techniques that seem inherently limited. Here we show for the first time that the maximum success probability of
 entangled games can be reduced through parallel repetition, provided it was not initially $1$.
Our proof is inspired by a seminal result of Feige and Kilian in the context of classical two-prover one-round interactive proofs.
One of the main components in our proof is an orthogonalization lemma for operators, which might be of independent interest.
\end{abstract}

\section{Introduction}

Two-player games play a major role both in theoretical computer science, where they have led to many breakthroughs such as
the discovery of tight inapproximability results for some constraint satisfaction problems, and in quantum physics, where they first arose in the context of Bell inequalities. In such games, a referee (or verifier) chooses a pair of questions from some distribution
and sends one question to each of two non-communicating players (or provers), who then respond with answers taken from some finite set. The
referee, based on the questions and answers, decides whether to accept (i.e., whether the players win). The main question of interest is the following: given the referee's
behavior as specified by the game, what is the maximum winning probability achievable by the players?
Somewhat surprisingly, the answer to this question turns out to depend on whether we force the players to behave
classically, or allow them to use quantum mechanics. In the former case,
the players' answers are simply deterministic functions of their inputs\footnote{One can allow the players to use randomness, but this does not change their maximum winning probability.}, and the maximum probability
of winning is known as the \emph{(classical) value} of the game.
In the latter case the players, though
still not allowed to communicate, may share an arbitrary entangled state and each perform arbitrary measurements on their share of the state.
The maximum winning probability in this case is known as the \emph{entangled value}
of the game. This model of entangled players (also known as that of \emph{non-local games}) dates back at least to the work of Tsirelson, and it
has been intensely studied in recent years; yet many questions about it are still wide open.

\medskip

One of the most important and interesting questions in this context is the parallel repetition question. It is well known that one can reduce both the value and the entangled value of a game by repeating it sequentially,
or alternatively, by repeating it in parallel with several independent pairs of players. However, for many
applications (like hardness of approximation results or amplifications preserving zero-knowledge) we need a way to decrease the winning probability without increasing the number of rounds or the number
of players, i.e., while staying in the model of two-player one-round games. Parallel repetition is designed
to do just that: in its most basic form, in the $\ell$-parallel repeated game, the referee simply chooses $\ell$
pairs of questions independently and sends to each player his corresponding $\ell$-tuple of questions. Each player then replies with an $\ell$-tuple
of answers, which are accepted if and only if each of the $\ell$ answer pairs would have been accepted in the original game.

Clearly the value of an $\ell$-parallel repeated game is \emph{at least} the $\ell$-th power of the value of the
original game, since the players can just answer each of the $\ell$ questions independently as in the original
protocol. However, contrary to what intuition might suggest and to the case of sequential
repetition, parallel repetition does \emph{not} necessarily decrease the value of a game in a straightforward
exponential manner\footnote{See~\cite{Fei91SCT} for a classical example, and~\cite{CleveSUU07} for an example using entangled
players due to Watrous. See also~\cite{Kempe2010a} for another example where parallel repetition does not reduce the value of a game at the exact rate one would expect if the players were playing independently.}. The parallel repetition question is that of finding \emph{upper bounds} on the value of a
repeated game, and for a long time no such upper bound, even very weak, could be proved. First results date to
Verbitsky~\cite{Verbitsky} who showed that indeed the value goes to zero with the number of repetitions. Following
this, Feige and Kilian~\cite{Feige2000} showed that the value decreases polynomially with the number of repetitions for
the special case of so-called \emph{projection games} (in which the second player's answer is uniquely determined by
the first player's). They used a modified parallel repetition procedure in which a large fraction of the repetitions
are made of \emph{dummy} rounds, that is, rounds in which the questions are chosen independently at random for both
players, and in which any answer is accepted. In this paper we deviate somewhat from the common terminology, and use
the term ``parallel repetition'' even when referring to such more general procedures. Finally, in a breakthrough
result, Raz~\cite{Raz98} showed that the value of a game repeated in parallel indeed decreases exponentially with the
number of repetitions (albeit not exactly at the same rate as sequential repetition). There is still very active
research in this area, mostly on simplifying the analysis, which, over a decade later, remains quite involved, and
improving it for certain special cases of
games~\cite{Holenstein07,Rao:parallel,FKR:foams,Raz08,BarakHHRRS08,BarakRRRS09,AroraKKSTV08,RazRosen10}. 

\medskip

\subsection{Previous work}

In this paper we focus on parallel repetition of \emph{games with entangled players}. The only two previous results in
this area are for two special classes of games. First, Cleve et al.\ showed that for the class of \emph{XOR games}
(i.e., games with binary answers in which the referee's decision is based solely on the XOR of the two answers), {\em
perfect} parallel repetition holds~\cite{CleveSUU07}. This means that the entangled value of an $\ell$-parallel
repeated game is exactly the $\ell$-th power of the entangled value of the original game. Parallel repetition has also
been shown to hold for the more general (but still quite restricted) class of \emph{unique games}~\cite{KempeRT:unique}
(i.e., games where the referee applies some permutation to the answers of the second player and accepts if and only
they match those from the first player). One might also add a third result by Holenstein~\cite{Holenstein07}, who
proved a parallel repetition theorem for the so-called \emph{no-signaling value}; since the no-signaling value is an
upper bound on the entangled value, this can sometimes be used to upper bound the entangled value of repeated games.
However, there is in general no guarantee regarding the quality of this upper bound, and in many cases (e.g., all
unique games) the no-signaling value is always $1$, making it useless as an upper bound on the entangled value.

It is important to note that in these results the entangled value of
the parallel repeated game is never analyzed directly; instead, one uses a ``proxy" such as
a semidefinite program~\cite{CleveSUU07,KempeRT:unique} or the no-signaling value~\cite{Holenstein07},
whose behavior under parallel repetition is well understood. Moreover, in all these
cases, the proxy's value is efficiently computable. This unfortunately gives a very strong indication that
such techniques cannot be extended to deal with general games. Indeed, it is known that
it is NP-hard to tell if the entangled value of a given game is $1$ or not~\cite{KempeKMTV2008,ItoKM09};
hence, unless P=NP, for any efficiently computable upper bound on the entangled value, there are necessarily
games whose entangled value is strictly less than $1$ yet for which that upper bound is $1$
(and such games can often be exhibited explicitly without relying on P$\ne$NP).
We note that some of the early parallel repetition results for the
classical value~\cite{FeiLov92STOC} followed the same route
(of upper bounding the value by a semidefinite program) and were limited
to special classes of games for the exact same reason.

\medskip

To summarize, no parallel repetition result (not even one with very slow decay) is known for the entangled
value of general games, and, moreover, the known techniques are unlikely to extend to this case.
Hence the natural question:

\begin{quote}
\emph{Can parallel repetition decrease the entangled value of games? If so, can we bound the rate of
decrease?}
\end{quote}

In parallel to work on the parallel repetition problem, the related question of {\em product testing} arose in the context of error amplification for PCPs~\cite{Dinur2006,Dinur2008,Impagliazzo2008,Impagliazzo2009}.
Roughly speaking, the question here is to design tests by which a referee can check that the players play according to a \emph{product
strategy}, i.e., answer each question independently of the other questions (as one would expect from an honest
behavior). Note that if the players are constrained to follow a product strategy, then their maximum winning probability
must necessarily go down exponentially, hence the connection to the parallel repetition question.
The result of Feige and Kilian~\cite{Feige2000} mentioned above in fact also shows that the strategy of the
players must have some product structure, and recently there has been lots of renewed interest in this question
leading to much stronger product testers~\cite{Dinur2010}.
 In the case of entangled players, however,
absolutely nothing was known:

\begin{quote}
\emph{Is there a way to test if the strategy of entangled players is in some sense close to a product strategy?}
\end{quote}

\subsection{Our results}

In this work we answer both questions in the affirmative, and our main result can be informally stated as follows.

\begin{theorem}[informal]\label{thm:maininformal}
For any $s<1$, $\delta>0$, and entangled game $G$, there is a corresponding $\ell$-parallel repeated game $G'$, where $\ell =
\poly((1-s)^{-1},\delta^{-1})$, such that if the value of $G$ is less than $s$ then the value of $G'$ is at most
$\delta$, whereas if the value of $G$ is $1$ then this also holds\footnote{See the discussion following
the theorem for some caveats.} for the repeated game.
\end{theorem}

The dependency of $\ell$ on $\delta$ in our theorem is polynomial, whereas as we already mentioned it is known that in some cases this dependence can be made poly-logarithmic (and this is certainly the case if the players are assumed to play independently). While a poly-logarithmic dependence is important in some applications for which one would like to perform amplification up to an exponentially small value, in many cases the main use of parallel repetition is to amplify a small ``gap'' between value $1$ and value $1-1/poly(|G|)$ to a constant gap, say between $1$ and $1/2$. In this case the polynomial dependence of $\ell$ on $(1-s)^{-1}$ that we obtain is optimal (up to the exact value of the exponent). 

In the course of the proof of this theorem we also establish that the player's strategies have a certain ``serial'' or ``product'' structure (more on this in the proof ideas and techniques section below). The informal statement above hides some details, which we now discuss.
The kind of parallel repetition we perform depends on the structure of the game $G$, and we distinguish whether it is a projection game or not.

\paragraph{Repetition for projection games.} If $G$ is a projection game, then the repeated game is obtained by independently playing the original
$G$ on a subset of the repetitions, and playing dummy rounds in the other repetitions. We note that projection games
form a wide class of games that captures most of the games one typically encounters in the classical literature
(see~\cite{Rao:parallel}).

If, in addition, the game happens to be a \emph{free} game (i.e., a game in which the
referee's distribution on question pairs is a product distribution), then the dummy questions are no longer
needed and hence our analysis applies to the
\emph{standard} $\ell$-fold repetition. 

\paragraph{Repetition for general games.} If the game $G$ does not have the projection property, then it is necessary to add a number of
\emph{consistency} rounds to the repetition. In those rounds the referee sends identical questions to the players, and expects identical answers. As before, the other rounds of the repetition are either the game $G$ or dummy rounds.  The consistency questions are added to play the role of the  projection constraints.

This kind of repetition raises the following issue\footnote{This is why we treat the projection case separately,
despite it leading to similar decay.}: namely, it is not obvious that honest entangled players can answer the consistency
questions correctly. This implies that, even if the original game had value $1$, players might not be able to succeed in the consistency questions and hence the value of the repeated game might not equal $1$ anymore. This may or may not be an issue depending on where the original game comes from. In many cases it is known that, if there is a perfect strategy, it does not require any entanglement at all, or it can be achieved using the maximally entangled state. In both cases it is not hard to see that players will be able to answer consistency questions perfectly, and hence our result holds. Because of this we regard this issue as a minor one; however it might be important in some contexts.

\subsection{Proof idea and techniques}

We focus on the case of projection games, as the proof of the other cases does not present additional challenges. The
starting point of our proof is the work of Feige and Kilian~\cite{Feige2000}, for which the following intuition can be
given\footnote{We refer to Ryan O'Donnell's excellent lecture notes~\cite{O'Donnell2005,O'Donnell2005a} for a helpful exposition of Feige and Kilian's proof.}. Our goal as the referee is to force the players to use a product strategy, preventing any elaborate cheating
strategies. In other words, we want to make sure that the player chooses his answer to the $i$th question based only on
that question and not on any of the other $\ell-1$ questions. Towards this end, the referee chooses a certain
(typically large) fraction of the $\ell$ question pairs to be independently distributed \emph{dummy questions}, the
answers to which are ignored. These dummy questions are meant to confuse the players: if they were indeed trying to
carefully choose their answer to a certain question by looking at many other questions, now most of these other
questions will be completely random and uncorrelated with the other player's questions, so that such a strategy cannot
possibly be helpful.

In more detail, Feige and Kilian prove the following dichotomy theorem on the structure of single-player repeated
strategies (that is, maps from $\ell$-tuples of questions to $\ell$-tuples of answers): either the strategy looks rather {\em random} (in which case the players cannot win the game with good
probability --- this is where the projection property is used) or it is almost a \emph{serial} or \emph{product}
strategy, i.e., the answer to each question is chosen based on that question only (in which case the player is playing
the rounds independently, and his success probability will suffer accordingly).

Our proof follows a similar structure. However, an important challenge immediately surfaces: the proof in~\cite{Feige2000}, and indeed \emph{all} proofs of parallel repetition theorems or direct product tests, make the important initial step of assuming that the player's strategies are deterministic (which is easily seen to hold without loss of generality). And indeed, it is not at all trivial to extend those proofs to even the randomized setting without making this initial simplifying assumption. To give a simple example, an important notion in Feige and Kilian's proof is that of a \emph{dead} question --- simply put, a question to which the player does not give any majority answer, when one goes over all possible ways of completing that specific question into a tuple of questions for the repeated game. It is easily seen that, in the case of a deterministic strategy, dead questions are harmful, as the players are unlikely to satisfy the projection property on them. However, it is just as easily seen that for most randomized strategies, good or bad, \emph{all} questions are dead.

This illustrates the kinds of difficulties that one encounters while trying to show parallel repetition in the case of entangled players, when one cannot simply ``fix the randomness''. The issue we just raised is not too hard to solve, and others are more challenging.
Indeed the main difficulty is to define a proper notion of {\em almost serial} for operators, which would in particular incorporate the inherent
randomness of quantum strategies. It turns our that the right notion is the notion of consecutive measurements (rather
than tensor products of measurements for each question, a tempting but excessively strong possibility).
Based on a quantum analogue of Feige and Kilian's dichotomy theorem, we are able to show that the almost serial condition induces a condition of {\em almost
orthogonality} on the player's operators. At this point we need to prove a genuinely quantum lemma, which lets us extract a \emph{product} strategy from the almost-orthogonal condition. This novel {\em orthogonalization lemma} is at the heart of our proof.
We obtain that the players approximately perform a series of consecutive measurements, each depending
only on the current question. An upper bound on the value of the repeated game then follows.

\paragraph{Organization of the paper.} We start with a few definitions, including a description of the form of the repeated games that we consider, in Section~\ref{sec:prelim}. We then give a high-level overview of the structure of the proof, and the main ideas governing it, in Section~\ref{sec:overview}. Section~\ref{sec:meproof} contains the proof of our main theorem. Finally, Section~\ref{sec:blockdiag} contains the proof of an important technical component of our proof: an approximate joint block-diagonalization of positive matrices which are close to being orthogonal. Appendix~\ref{sec:exp} contains a few additional useful technical facts.


\section{Preliminaries}\label{sec:prelim}

\subsection{Games}\label{sec:games}

In this paper we study two-player one-round games.
Let $Q$ and $A$ be finite sets. An entangled game (or simply game) can be defined as follows.

\begin{definition} An entangled game $G = (V,\pi)$ is given by a function $V \colon A^2
\times Q^2 \rightarrow \{0,1\}$ and a distribution $\pi \colon Q^2 \to [0, 1]$. The referee samples questions $(q',q)$ according to $\pi$, and sends $q'$ to the first player and $q$ to the second player. He receives back answers $a',a$ respectively. He accepts those answers if and only if
$V(a',a \mid q',q)=1$.  The value of the game is
 $$\omega^* (G) = \sup_{\ket{\Psi},A_q,B_q}
 \,\sum_{(q', q) \in Q^2}\,\sum_{ (a',  a) \in A^2}\,
\pi(q',q) V(a',a|q',q)\, \bra{\Psi} A_{q'}^{a'} \otimes B_{q}^a \ket{\Psi}$$
where the supremum is taken over all finite-dimensional Hilbert spaces $\mathcal{H}$, all a priori shared states $\ket{\Psi}\in\mathcal{H}$ and all Projective Operator-Valued Measurements (POVMs)\footnote{The POVM condition states that each $A_{q'}^{a'} \geq 0$, and $\sum_{a'} A_{q'}^{a'} = \Id$.} $A_{q'} =
\{A_{q'}^{a'}\}_{a' \in A}$ and  $B_{q} = \{B_{q}^{a}\}_{a \in A}$ on $\mathcal{H}$. 
\end{definition}

We note that by standard purification techniques (see~\cite{Cleve2004}) one can assume that for each question $q$ each player performs a projective
  measurement with outcomes in $A$ (i.e., $\sum_{a \in A} A^a_q=\Id$ and
  $(A_q^a)^\dagger=A_q^a=(A_q^a)^2$).

   We will be interested in some special classes of games.
  
  \begin{definition} A game $= (V,\pi)$ is called a 
  \begin{itemize}
  \item \emph{Projection game} if for every $q',q\in Q$ and $a' \in A$, there is a unique $a\in A$ such that $V(a',a|q',q)=1$.
  \item \emph{Free game} if $\pi = \pi_A \times \pi_B$ is a product distribution.
  \item \emph{Symmetric game} if $\pi$ is symmetric, and for any $q',q,a',a$ we have $V(a',a|q',q) = V(a,a'|q,q')$.
  \end{itemize}
\end{definition}

\subsection{Repeated games}

We consider two different types of repeated games. The first one, originally used by Feige
and Kilian, applies to projection games, and we describe it in Definition~\ref{def:fkrep}. The second type of
repetition applies to consistency games, and is closer to the direct product testing technique originally introduced by
Dinur and Reingold~\cite{Dinur2006}; we explain it in Definition~\ref{def:drrep}.

\begin{definition}[Feige-Kilian repetition]\label{def:fkrep}
Let $\ell$ be any integer, and define $C_1:=\ell^{1/2}$ and $C_2:=\ell-C_1$. Given a two-player projection game $G = (\pi,V,Q,A)$, its $\ell$-th Feige-Kilian repetition is the following game $G_{FK(\ell)}$:
\begin{itemize}
\item The referee picks a random partition $[\ell] = M \cup F$, where $|M|=C_1$ and $|F|=C_2=\ell-C_1$. Indices in $M$ will be
called ``game'' indices, while indices in $F$ will be called ``confuse'' indices.
\item The referee picks $(q_M',q_M) \sim_{\pi^{C_1}} (Q\times Q)^{C_1}$.
 \item He picks $(q_F',q_F)
\sim_{(\pi_A\times\pi_B)^{C_2}} (Q\times Q)^{C_2}$, where $\pi_A$ is the marginal of $\pi$ on the first player, and
$\pi_B$ the marginal on the second player. \item The referee sends the questions to the players (without specifying
which questions are of which type). On game questions he verifies that the original game constraint is satisfied. He
accepts any answers to confuse questions.
\end{itemize}
\end{definition}

\begin{definition}[Dinur-Reingold repetition]\label{def:drrep}
Let $\ell$ be any integer, and define $C'_1:=\ell^{1/2}$, $C_1=2C'_1$ and $C_2:=\ell-C_1$. 
Given a two-player symmetric game $G=(\pi,V,Q,A)$, its $\ell$-th Dinur-Reingold repetition is the following game
$G_{DR(\ell)}$:
\begin{itemize}
\item The referee picks a random partition $[\ell] = R\cup G \cup F$, where $|R|=C'_1$, $|G|=C'_1$,
and $|F|=C_2$. Indices in $R$ will be called ``consistency'' indices, those in
$G$ will be called ``game'' indices, and those in $F$ ``confuse'' indices. 
\item The referee picks $C'_1$ questions
$q_{R}\sim_{\pi_A^{C'_1}} Q^{C'_1}$ and sets $q_R' = q_R$, where $\pi_A$ is the marginal of $\pi$ on the first player
(since we assumed $G$ was symmetric, this is the same as $\pi_B$, the marginal on the second player). \item The referee
picks $C'_1$ pairs of questions $(q_{G}',q_{G})\sim_{\pi^{C'_1}} (Q\times Q)^{C'_1}$.
\item He picks $(q_F',q_F)
\sim_{(\pi_A\times\pi_B)^{C_2}} (Q\times Q)^{C_2}$.
 \item The referee sends the questions to the players (without
specifying which questions are of which type). On consistency questions he verifies that both answers, from Alice and
from Bob, are identical. On game questions he verifies that the original game constraint is satisfied. He accepts any
answers to confuse questions.
\end{itemize}
\end{definition}

Note that, if a game $G$ has value $1$, then its Dinur-Reingold repetition does not necessarily also have value $1$, as the player's optimal strategy in $G$ might not be \emph{consistent}. A consistent strategy is one in which whenever the players are asked the same question they provide the same answer with certainty. This may not always hold of an optimal strategy; nevertheless the following lemma shows that we can assume it holds in some natural settings.

\begin{lemma}[Lemmas~3 and~4 in~\cite{KempeKMTV2008}]\label{lem:consistentplayers}
Let $G=(V,\pi)$ be an arbitrary $2$-player entangled game. Then there exists a game $G'=(V',\pi')$ \emph{of the same classical and quantum values} with twice as many questions, and such that $\pi'$ and $V'$ are
symmetric under permutation of the variables. Moreover, given
any strategy $P_1,\ldots,P_N$ with entangled state $\ket{\Psi}$ that wins $G$ with probability $p$, there exists
a strategy $P'_1,\ldots, P'_N$ with entangled state $\ket{\Psi'}$ that wins $G'$ with probability $p$ and is such that
$P'_{1}=\cdots = P'_{k}$ and $\ket{\Psi'}$ is symmetric with respect to the provers $1,\ldots,k$. In addition, if $\ket{\Psi}$ was a maximally entangled state then $\ket{\Psi'}$ is also. 
\end{lemma}

This lemma shows that, if $G$ is any game, then we may symmetrize it and assume that the provers are also playing according to a symmetric strategy. In particular, if $G$ had value $1$, and the optimal strategy used either no entanglement or a maximally entangled state, then this also holds of the optimal strategy in the symmetrized game. Such a strategy is automatically consistent.


\section{Proof overview}\label{sec:overview}

We first give a formal account of our results in the next section, before proceeding to give an overview of their proof in Section~\ref{sec:pfoverview}.

\subsection{Results}
We first state our main theorems. They refer to the two types of repetition of an entangled game $G$ defined in the previous section,
its $\ell$-th \emph{Feige-Kilian repetition} $G_{FK(\ell)}$, and its $\ell$-th \emph{Dinur-Reingold} repetition $G_{DR(\ell)}$. Both types of repeated games are made of $\ell$ independent rounds, played in parallel. Some of these rounds
consist of independent repetitions of $G$, while others are either \emph{confuse} or \emph{consistency} rounds, containing
simple tests independent of the original game (except for the distribution with which questions are chosen in those
rounds). Our first result pertains to projection games.

\begin{theorem}\label{thm:mainproj}
There exists a constant $c\geq 1$ such that, for all $s<1$ and $\delta>0$ there is a $\ell = O((\delta^{-1}\,(1-s)^{-1})^{c})$ such that, if 
 $G$ is a projection game with value $\omega^*(G) \leq s$,
then the entangled value of the game $G_{FK(\ell)}$ is at most $\delta$. Moreover, if the value of $G$ is $1$ then the
value of $G_{FK(\ell)}$ is also $1$.
\end{theorem}

In the case of free projection games, questions to the players are chosen independently, so that the distribution on questions in the confuse rounds of the game $G_{FK(\ell)}$ is exactly the same as that in the original game. The only difference is that in such a round, all answers are accepted, which can only help the players. Hence the direct parallel repetition of $G$ has a smaller value than its Feige-Kilian repetition, which implies the following.

\begin{corollary}
Let $s<1$ and $\delta>0$. Then there is a $\ell = O((\delta^{-1}\,(1-s)^{-1})^{c})$ such that, if $G$ is a free projection game such that $\omega^*(G) \leq s$, then the (direct) $\ell$-fold parallel repetition of $G$ has value at most $\delta$. 
\end{corollary}

Our second result is more general, as it applies to arbitrary games. It only comes with the mild caveat that, in order to preserve the fact that the original game had value $1$ (whenever this indeed holds), it is required that in that case there also exists a perfect strategy which is consistent. 

\begin{theorem}\label{thm:maincons}
There exists a constant $c\geq 1$ such that, for all $s<1$ and $\delta>0$ there is a $\ell = O((\delta^{-1}\,(1-s)^{-1})^{c})$ such that, if 
 $G$ is an arbitrary game with value $\omega^*(G)\leq s$ , then
 the entangled value of the game $G_{DR(\ell)}$ is at most $\delta$. Moreover, if $G$ has a perfect consistent strategy then the value of
$G_{DR(\ell)}$ is also $1$.
\end{theorem}

Lemma~\ref{lem:consistentplayers} shows that the requirement that $G$ has a perfect consistent strategy (which is only a requirement in cases where we are interested in preserving the fact that $G$ might have value $1$) is satisfied for many examples of games, including those for which we know a priori that, if the value of $G$ is $1$, then there is an optimal strategy that either does not use any entanglement at all, or uses the maximally entangled state.

\subsection{Proof overview}\label{sec:pfoverview}

In the remainder of this section we describe the main ideas behind the proof of Theorem~\ref{thm:mainproj} and Theorem~\ref{thm:maincons};
full details can be found in Sections~\ref{sec:meproof} and~\ref{sec:blockdiag}. 
Our goal is to understand
\emph{repeated} quantum strategies, that is, maps $q\in Q^\ell \mapsto \{\,X_q^a\,\}_{a\in A^\ell}$ which map tuples of
questions $q=(q_1,\ldots,q_\ell)$ to projective measurements $\{\,X_q^a\,\}_{a\in A^\ell}$ in dimension $d$. The semantics are that, on receiving questions
$q$, a player measures his share of the entangled state $\ket{\Psi}$ according to $\{\,X_q^a\,\}_{a\in A^\ell}$,
resulting in him sending back answer $a$ with probability $\bra{\Psi} \Id \otimes X_q^a \ket{\Psi}$. Interestingly, most of the proof will be directly concerned with
the measurements $\{\,X_q^a\,\}_{a\in A^\ell}$ themselves (together with the reduced density $\rho = \Tr_A \ket{\Psi}\bra{\Psi}$), without reference to the other player's
measurements or even the underlying game.

We will be interested in a strategy's \emph{marginals}: given a fixed subset of indices $S\subseteq [\ell]$ and a set of questions $q_S$ on the indices in $S$, one can define the marginalized measurement
$$\Big\{ X_{q_S}^{a_S} :\,\rho\,\mapsto\,  \textsc{E}_{q\sim \pi^{[\ell]\backslash S}}\Big[\,\sum_{a\in A^{[\ell]\backslash S}} \,\sqrt{X_{q_Sq }^{ a_S\,a }}\,\rho\, \sqrt{X_{q_Sq }^{ a_S\,a }}\,\Big]\,\Big\}_{a_S \in A^{S}}$$
which corresponds to choosing a tuple $q\in Q^{[\ell]\backslash S}$ by picking the question in each coordinate independently according to some fixed distribution $\pi$,\footnote{We will often drop the reference to $\pi$ and simply write $\Exs{q}{\cdot}$. $\pi$ will be fixed throughout, and later instantiated to the (marginal) distribution on questions from the original game $G$ that is being repeated.}  making the measurement corresponding to the POVM described by $\{\,X_{q_S q}^{a_S a}\,\}_{(a_S, a)\in A^\ell}$, and marginalizing over those answers $a$ corresponding to indices not in $S$.
 
Given that $X$ was a projective measurement, the marginalized strategy is a POVM --- it is not necessarily projective
any more. Our main results will pertain to the structure of such marginalized strategies. We will show that they are either very random (this is formally called \emph{dead} later on, and morally means that the
marginalized strategy is very far from a projective measurement; rather its singular values tend to be small and
spread out), or highly structured (this is called \emph{serial} later on, and after some work we will show that it
implies that the marginalized strategy has somewhat of a product form, i.e.~it can be decomposed as a product
$\Pi_{q_1}^{a_1}\cdots \Pi_{q_{s}}^{a_{s}}$ on a subset of the coordinates). The attentive reader might already see
that once this is proven it will be possible to bound the success probability of both types of strategies in the repeated game; however we
should warn that the exact statements, and their proofs, are quite technical and carry only a fair share of the
intuition we have just given.

\medskip

We proceed to give a few more details on the structure of the proof of our results. It can be divided into three main
steps. The first two steps establish facts about the structure of repeated single-player strategies, and are independent of
the game being played, as well as of the other player's strategy. 

\paragraph{Step 1: A quantum dichotomy theorem.}
In the first step we prove an analogue of Feige and Kilian's
dichotomy theorem~\cite{Feige2000}. The precise statement is given in Lemma~\ref{lem:alive}, and its simple proof very
closely follows that of Feige and Kilian's theorem. Informally, it states that there exists an integer $1\leq r^* \ll
\ell$, such that a tuple of questions $(R,q_R)$, where $R \subseteq [\ell]$ denotes a subset of $r^*$ indices, and
$q_R$ fixed questions in those positions, can be of two types only. Either it is \emph{dead} (case 1 in the lemma), or
it is $(1-\eta)$\emph{-serial}, where $\eta>0$ is a small parameter (case 2 in the lemma). Both types of strategies are precisely defined in Definition~\ref{def:serial}, and the meaning of dead is the easiest to grasp. The
technical definition is simply that the (marginalized) measurement $\{X_{q_R}^{a_R}\}_{a_R\in A^R}$, when performed
twice (sequentially) on the same half\footnote{In fact we will also need to consider the outcome of performing the same measurement simultaneously on the two halves of $\ket{\Psi}$.} of the state $\ket{\Psi}$, is unlikely to produce the same result. This kind of strategy is easily seen to be bad for the players, as is shown in step 3. of the proof. 

Serial strategies are more subtle. In the case of a classical deterministic player, a serial strategy is such that, when
one conditions on the player giving answers $a_R$ to the questions $q_R$ in $R$, the answers to most other questions (not in $R$) are for the
most part determined by the player as a direct function of the corresponding question, i.e. he is playing an honest
product strategy on those coordinates. In the quantum case, we will adopt a seemingly weaker definition, which is that
a strategy is serial on $(q_R,a_R)$ if, in expectation over the choice of an additional question $q_i$ in position $i$, when the marginalized measurement $\{X_{q_R q_i}^{a_R a_i}\}_{(a_R,a_i)\in A^{R\cup\{i\}}}$ is performed twice on the same half of $\ket{\Psi}$, the probability that the same answer $(a_R,a_i)$ is obtained twice is almost as large as the probability that just $a_R$ is obtained twice: conditioned on being consistent on the answers to the questions in $R$, the strategy is also consistent in its answer on a random additional question $q_i$ in position $i$. 

Fleshing out the consequences of this definition to eventually show that it implies something close to the classical definition requires some work, and is the object of the second step of the proof.

\paragraph{Step 2: A product theorem for serial strategies.}
While for a classical deterministic player a serial strategy, as defined in the previous section, is one which decides on the answer $a_i$ to most questions $q_i$ not in $R$ as a function of that question alone, in the quantum setting this is much less clear. The first task is to decide on what one expects from a serial strategy. For instance, one might ask for the measurements to take some ``approximately-tensor''
form; however we find that this is too strong a requirement. Instead, we first show that the serial property implies that the player's measurement operator $\{X_{q_R q_i}^{a_R a_i}\}_{(a_R,a_i)\in A^{R\cup\{i\}}}$ has a certain block-diagonal form, in the sense that\footnote{Note that this ``approximation'' should be taken with a grain of salt; in particular one cannot expect to extract any information about the measurement operators themselves simply by observing statistics of measurement outcomes. Rather, all our estimates will bear on the post-measurement state, resulting from applying the measurement corresponding to $ X_{q_R q_i}^{a_R a_i}$ to one half of $\ket{\Psi}$.}
$$X_{q_R q_i}^{a_R a_i} \,\approx \,\Pi_{q_i}^{a_i}\,X_{q_R q_i}^{a_R a_i}\, \Pi_{q_i}^{a_i}$$
where $\{\Pi_{q_i}^{a_i}\}_{a_i\in A}$ are \emph{orthogonal} projectors; the precise statement is given in Claim~\ref{claim:blocksgen}. Its proof goes through a technical statement about sets of operators which are close to being pairwise orthogonal. That statement, proven in Lemma~\ref{lem:blockdiaggen}, shows the natural fact that such operators are close to having a common block-diagonalization basis.

Once this is shown it is not hard to extend the approximation to a small number of additional questions $q_1,\ldots,q_g$, showing that the corresponding measurement also has a block-diagonal form, this time described by the product of the corresponding projectors $\Pi_{q_1}^{a_1} \cdots \Pi_{q_g}^{a_g}$; a precise statement is given in Lemma~\ref{lem:productgen}.
It is in the precise sense described in that lemma that we can say that a serial strategy has a product form, based on which we can think of the player as playing sequentially on a subset of the coordinates. 

\paragraph{Step 3: Both dead and serial strategies fail the repeated game.} 
In the last step of the proof we show that both types of strategies, dead or serial, must fail in the repeated game with high probability (provided the value of the original game was bounded away from $1$). For the case of dead strategies this is fairly intuitive: since a dead strategy does not assign consistent answers to a certain  subset of the questions $q_R$, this implies that the player's answers in positions $R$ will very much depend on the
questions present in those indices not in $R$; not only that but it will be virtually impossible for the other player
to correlate well with this player's answers on those indices. Here we crucially use the ``projection'', or ``consistency'' rounds of the repeated game in order to show that such strategies will fail in those rounds with high probability. This is proven in Claim~\ref{claim:deadsuccess}. 

The case of serial strategies is slightly harder to analyze, but it boils down to showing that the block-diagonal form we described earlier roughly implies that we can in fact see one of the players as making a sequential measurement governed by the $\Pi_{q_i}^{a_i}$. Since in this case the player's answer to question $q_i$ is decided by applying a projective measurement depending on $q_i$ alone, in case the original game had a value $s<1$ such a strategy will fail in at least a fraction $s/2$ of the ``game'' rounds with high probability, and be caught by the referee provided there are enough such rounds. This is shown in Claim~\ref{claim:prodsuccess}.

Finally note that the ``confuse'' rounds of the repeated game are not used in this stage (and indeed the referee accepts any answers in those rounds), but they are crucial to show the dichotomy lemma and the following claims, which only hold for strategies which have been marginalized over a sufficiently large number of questions; in order to be able to perform this marginalization it is important that questions to the players in the confuse rounds are picked independently.


\section{Proof of the main theorem}\label{sec:meproof}

In this section we give the proof our main results, Theorems~\ref{thm:mainproj} and~\ref{thm:maincons}. It is divided in three parts. The first, in Section~\ref{sec:dicho}, establishes our ``quantum dichotomy theorem''. The second, in Section~\ref{sec:directprod}, investigates the structure of serial strategies, and shows that they admit a certain block structure. The results in this section are based on our ``orthogonalization lemma'', which is proved separately in Section~\ref{sec:blockdiag}. Finally, in the third part, Section~\ref{sec:succbound}, we use the results from the first two parts to bound the success probability of the players in the repeated game. 

Because of the nature of repeated strategies, which are indexed by large tuples of questions and answers, we are constrained to use rather heavy notation. We explain it in detail in the following section, which can also serve as a reading guide for the statements that are to follow. 

\subsection{Notation}\label{sec:notation}

Recall that for every $q\in Q^\ell$, $\{X_q^a\}_{a\in A^\ell}$ is an arbitrary projective measurement in $d$
dimensions, that is, the $X_q^a$ are projector matrices, and for any fixed $q$ they sum to the identity over $a$. 
The position of the questions (or answers) in a tuple will always be fixed and usually clear from the
context; for example when we write $q = (q_G,q_F)$, where $G,F \subseteq [\ell]$ are sets of indices, it is not necessary that the
questions $q_G$ are placed before the questions $q_F$ in the tuple $q$; rather their position is determined by the
indices in $G,F$. When precision is needed we shall write $(i,q_i)$ to express the fact that question $q_i$ is destined
to appear in the $i$-th position of some tuple $q$. We also write $q_{\neg i}$ to denote $q_1,\ldots ,q_{i-1},q_{i+1},\ldots
,q_\ell$.

 We will often consider \emph{marginalized} POVMs over a certain set $S \subseteq [\ell]$. Given questions $q_S$ indexed by $S$, the marginalized POVM is the POVM indexed by answers $a_S$, which results from applying $\{X_{q_S q}^{a_S a}\}_{a_S a}$ for a random $q \in Q^{[\ell]\backslash S}$, and ignoring the answers $a$ not in $S$. More precisely, given $(S,q_S,a_S)$ it will be convenient to work with the Stinespring representation 
$$\hat{X}_{q_S}^{a_S} := \sum_q\sum_a\, \sqrt{\pi(q)}\, \sqrt{X_{q_S q}^{a_S a}} \otimes \bra{q,a}_E $$
where $E$ is an extra register of the appropriate dimension, and $\pi$ denotes an arbitrary distribution, fixed throughout (it will later be instantiated to the marginal distribution that arises from the original game $G$ that is being repeated). This definition satisfies, for any $\rho \geq 0$, 
$$ \textsc{E}_q \Big[\,\sum_a\, \sqrt{X_{q_S q}^{a_S a}} \,\rho\, \sqrt{X_{q_S q}^{ a_S a}} \,\Big] =  \hat{X}_{q_S}^{a_S} \,(\rho\otimes \Id_E)\, (\hat{X}_{q_S}^{a_S})^\dagger $$
where the identity $\Id_E$ was created on the additional register $E$ introduced in the definition of $\hat{X}_{q_S}^{a_S}$, and the expectation is with respect to the distribution $\pi$. In order to make measurements corresponding to marginalization over different sets $S$, we will assume that the register $E$ is always of large enough dimension, and if necessary $\hat{X}_{q_S}^{a_S}$ is tensored with $\frac{1}{\sqrt{|Q|^{|S|} |A|^{|S|}}} \sum_{q,a} \bra{q,a}$ on the extra $2|S|$ registers. Note that there is nothing in the definitions above that require the questions and answers in $\hat{X}_{q_S}^{a_S}$ to be indexed to the same set, hence we extend them to define $\hat{X}_{q_S}^{a_T}$, for $T\subseteq S\subseteq [\ell]$, in the obvious way. 

For any $\rho \geq 0$, we write $\Tr_\rho(A)$ for $\Tr(A(\rho\otimes\Id_E))$, so that in particular 
$$\Trho\big((\hat{X}_{q_S}^{a_T})^\dagger \hat{X}_{q_S}^{a_T}\big) \,=\, \textsc{E}_q \Big[\,\sum_a\,  \Tr\big(\sqrt{X_{ q_S q}^{a_T a}} \,\rho\, \sqrt{X_{ q_S q}^{a_T a}}\big)\,\Big]\,=\, \Tr\big( X_{q_S}^{a_T}\rho\big)$$
where we define
$$X_{q_S}^{a_T}  \,:=\, \hat{X}_{q_S}^{a_T} (\hat{X}_{q_S}^{a_T})^\dagger \,=\,\textsc{E}_{q\in Q^{[\ell]\backslash S}} \Big[\,\sum_{a\in A^{[\ell]\backslash T}}\, X_{q_S q}^{a_T a}\,\Big]$$
Terms such as $\Tr\big( X_{q_S}^{a_T}\rho\big)$ will frequently appear on the right-hand side of our inequalities, and they should simply be considered as normalization factors, accounting for the (possibly unnormalized) underlying state $\rho$, and the conditioning on a fixed $a_T$. Finally, given $\rho\geq 0$ and a matrix $A$ of appropriate dimension, we introduce the semi-norm
\beq\label{eq:seminorm}
 \|A \|_\rho^2\,:=\, \Tr\big( A \rho^{1/2} A^\dagger \rho^{1/2} \big) 
 \eeq
Note that $\|\cdot \|_\rho$ is definite only if $\rho$ has full rank. We will mostly use this norm for notational convenience. At this point it suffices to observe that it derives from a semi inner-product, so that it satisfies the Cauchy-Schwarz inequality. 

At a first reading it may be helpful for the reader to consider the special case of the totally mixed state $\rho = d^{-1} \Id$; putting the notation in context this corresponds to the players sharing the maximally entangled state. In this case very little of the above is really needed, and in particular $\Trho\big((X_{q_S}^{a_T})^\dagger X_{q_S}^{a_T}\big)$ is simply the normalized trace $\Exs{q}{\sum_{a} d^{-1}\Tr\big( X_{q_S q}^{a_T a} \big)}$. Many of our statements are easier to prove, and to understand, in this setting (the main cause of simplification being the commutation between $\rho$ and the $X$ operators), so that the reader may wish to consider it first. 

\subsection{A quantum dichotomy theorem}\label{sec:dicho}

In this section we prove two important lemmas on the structure of any quantum strategy in a repeated game.
The main lemma, Lemma~\ref{lem:alive}, is the analogue of Lemma~11 in~\cite{Feige2000}. It establishes a
dichotomy between two different types of strategies that a player can use, showing that either the strategy is very
random, or it must have a relatively strong sequential structure. Its proof follows that of the classical
setting without too much added difficulty, provided the definitions are made correctly --- which we now proceed to do. 

A crucial difficulty in adapting Feige and Kilian's argument is to define an appropriate measure of a strategy's \emph{unpredictability}. In the classical case of a deterministic strategy, this can be measured through the entropy of the marginalized distribution on answers; however in the quantum or even the randomized setting such a measure is no longer helpful, as even honest product strategies can be very random, just by being convex combinations of distinct deterministic strategies. Instead, we measure unpredictability as follows.

\begin{definition} Given a strategy $X_q^a$, a state $\rho$, and a fixed set of questions $q_R$ in positions $R\subseteq [\ell]$, define the collision probability of $X$ on $q_R$ as 
\beq\label{eq:pcol}
P_\text{col}(q_R|X,\rho) := \sum_{a_R} P_\text{col}(q_R,a_R|X,\rho)
\eeq
 where
\beq P_\text{col}(q_R,a_R|X,\rho) \,:=\, \Big(\Trho\big( (\hat{X}_{q_R}^{a_R})^\dagger\,\hat{X}_{q_R}^{a_R}(\hat{X}_{q_R}^{a_R})^\dagger\, \hat{X}_{q_R}^{a_R}\big)+\Tr\big( {X}_{q_R}^{a_R}\,\rho^{1/2}{X}_{q_R}^{a_R}\,\rho^{1/2}\big)\Big)\label{eq:pcol2}
\eeq
\end{definition}

To understand this definition, first consider the case when $\rho$ is the totally mixed state $d^{-1}\Id$. In this case both terms inside the summation are equal to the normalized squared Frobenius norm $d^{-1} \|X_{q_R}^{a_R}\|_F^2$. Expression~\eqref{eq:pcol} can be interpreted in two different ways. From an operational point of view, it corresponds to the probability that one obtains twice the same answers when one sequentially performs a measurement using the POVM with elements $\{X_{q_R}^{a_R}\}_{a_R}$. In this sense, $P_\text{col}$ is a measure of the predictability of the strategy $X_q^a$: pick two completions $q,q'$ at random and measure using first $\{X_{q_R q}^{a_Ra}\}_{a_Ra}$ and then using $\{X_{q_R q'}^{a_Ra'}\}_{a_Ra'}$; $P_\text{col}(q_R|X,\rho)$ is the probability of getting twice the same result $a_R$ (and ignoring the other answers $a,a'$). The analytic interpretation is that this is a measure of the entropy of the spectrum of $X_{q_R}^{a_R}$, which is maximized when $X_{q_R}^{a_R}$ is a projector (for a fixed value of the trace).

In case $\rho$ is not the identity, and hence does not commute with the $X_q^a$, we need to adopt the more cumbersome definition~\eqref{eq:pcol} for technical reasons. However, note that the operational interpretation remains --- the first term on the right-hand side of~\eqref{eq:pcol2} is the probability of obtaining the same answer when performing the measurement twice on the \emph{same half} of $\ket{\Psi}$, while the second term is the same, when the measurement is performed on the two \emph{different halves} of $\ket{\Psi}$: indeed, note that $\Tr\big( {X}_{q_R}^{a_R}\,\rho^{1/2}{X}_{q_R}^{a_R}\,\rho^{1/2}\big) = \bra{\Psi} X_{q_R}^{a_R}\otimes (X_{q_R}^{a_R})^T \ket{\Psi}$.\footnote{Note the transpose sign, which indicates that our interpretation is only rigorously correct for the case of real symmetric $X$.}

\medskip

The following lets us make the distinction between the two different types of strategies alluded to above.

\begin{definition}\label{def:serial} We will say that:
\begin{itemize}
\item A block $(R,q_R)$ is $\eps$-\emph{dead} if $P_\text{col}(q_R|X,\rho)\leq\eps$.
If a block is not $\eps$-dead it is $\eps$-\emph{alive}. Moreover, we say that the answer $a_R$ is $\eps$-\emph{alive} if it satisfies 
$$P_\text{col}(q_R,a_R|X,\rho)\, \geq \,\eps\, \Tr\big( X_{q_R}^{a_R}\rho\big)$$
 Note that any $\eps$-alive block has at least one $\eps$-alive answer. Sometimes we will simply say that a block or an answer are alive or dead, leaving the parameter $\eps$ implicit.
\item A block $(R,q_R,a_R)$ is $(1-\eta)$-serial if $a_R$ is alive and the following holds:
\beq
 \Exs{(i,q_i)}{\,P_\text{col}(q_{R},q_i|X,\rho)\,}  \,\geq\, (1-\eta) P_\text{col}(q_R|X,\rho)\label{eq:serial}
\eeq  
\end{itemize}
\end{definition}

\begin{lemma}\label{lem:alive} Assume that $\eps,\eta>0$ are chosen such that $\eta\,\eps^3 > 16\,C_1^{-1/2}$.\footnote{Recall that $C_1,C_2$ are chosen so that $C_1 + C_2 = \ell$: see Definitions~\ref{def:fkrep} and~\ref{def:drrep} for more details.} Then one of the following holds
\begin{enumerate}
\item At least a $(1-\eps)$ fraction of blocks $(R,q_R)$ are $\eps$-dead.
\item At least an $\eps$ fraction of blocks $(R,q_R)$ are $\eps$-alive, and moreover if $(R,q_R)$ is an $\eps$-alive block then
\beq\label{eq:alivenotserial}
\sum_{a_R:\,a_R\text{ alive but }\atop (q_R,a_R)\text{ is not $(1-\eta)$-serial}} \Tr\big( X_{q_R }^{a_R}\rho\big) \,\leq\,\eps/2
\eeq
i.e. alive answers which are not $(1-\eta)$-serial have a small probability of occurring.
\end{enumerate}
\end{lemma}

\begin{proof}
We extend the definition of the collision probability to measuring collisions over answers which are not necessarily on the same indices as the questions: 
$$P_\text{col}(q|R,X,\rho) \,:=\,\sum_{a_R} \Big(\Trho\big( (\hat{X}_{q}^{a_R})^\dagger\,\hat{X}_{q}^{a_R}(\hat{X}_{q}^{a_R})^\dagger\, \hat{X}_{q}^{a_R}\big)+\Tr\big( X_{q}^{a_R}\,\rho^{1/2}X_{q}^{a_R}\,\rho^{1/2}\big)\Big)$$
where now $q$ can be any subset of fixed questions, and $R$ denotes the subset of answers on which we are measuring the collision probability. 

\begin{claim}\label{claim:pred}
There exists an integer $1\leq r^* \leq C_1$ such that
$$ \Exs{R,q_R}{\,P_\text{col}(q_R|R,X,\rho)\,} - \Exs{R,q_R,i,q_i}{\,P_\text{col}(q_R,q_i|R\cup\{i\},X,\rho)\,} \leq 8\, C_1^{-1/2}$$
where the expectation is taken over all subsets $R$ of size $|R| = r^*$.
\end{claim}

\begin{proof}
There is a similar statement in~\cite{Feige2000}. Here we closely follow the proof of Corollary~3.2 in the lecture notes~\cite{O'Donnell2005}; since the argument is very similar (mostly replacing the use of Fact~1.3 in those notes by our Claims~\ref{claim:exptrace3} and~\ref{claim:exptrace2}) we only outline it here, leaving the details to the reader. The proof goes by considering what happens to the collision probability when one conditions on an additional question, resp. one considers collisions over an additional answer. First, note that if one extends $R$ by an index $i$, then $P_\text{col}(q|R\cup\{i\},X,\rho) \leq P_\text{col}(q|R,X,\rho)$, since obtaining identical answers on $R$ is a necessary condition to obtain identical answers on $R\cup\{i\}$. The following equation is the analogue of Fact~1.4 in~\cite{O'Donnell2005}:
\beq\label{eq:predinc}
\big|\Exs{(i,q_i)}{ P_\text{col}(q,q_i|R,X,\rho)} - P_\text{col}(q|R,X,\rho) \big|\,\leq\, 4\,C_1^{-1/2}
\eeq
The proof of~\eqref{eq:predinc} follows directly from Claims~~\ref{claim:exptrace3} and~\ref{claim:exptrace2}, and we omit it.
It shows that the collision probability cannot increase by too much when one conditions on an additional question, in expectation. The proof of the claim is then concluded exactly as in the classical case: consider a sequence of steps in which one successively looks for collisions on an additional coordinate $i$, and conditions on an additional question $q_i$. In expectation over the choice of $(i,q_i)$, $P_{\text{col}}$ will never go up by more than $4C_1^{-1/2}$ when one performs this operation. Since $P_{\text{col}}$ is always between $0$ and $1$, the fact that it never goes up by much implies that there must be a step in which it doesn't decrease by more than $8C_1^{-1/2}$: the total decrease cannot be larger than the total increase plus $1$. $r^*$ is chosen so that this step occurs when $r^*$ indices (and questions) have already been fixed.  
\end{proof}

Towards a contradiction, assume the negation of both 1. and 2. With probability at least $\eps$ a random block $(R,q_R)$ is alive, and moreover if $(R,q_R)$ is alive then alive answers which are not $(1-\eta)$-serial have a significant contribution.
Fix such an answer $a_R$. Since~\eqref{eq:serial} is not satisfied, summing over all $a_R$ which are alive but not $(1-\eta)$-serial one can see that the collision probability, for this $(R,q_R)$, must decrease by at least
$$ \eta\cdot\sum_{a_R:\,a_R\text{ alive but }\atop (q_R,a_R)\text{ is not $(1-\eta)$-serial}} P_\text{col}(q_R,a_R|X,\rho)$$
By the negation of~\eqref{eq:alivenotserial} and the fact that the answers are alive, this quantity is at least $\eta \eps^2/2$. 
 Finally, taking the expectation over the choice of $(R,q_R)$ gives a total decrease in $P_{\text{col}}$ of
at least $\eta \eps^3/2$, contradicting Claim~\ref{claim:pred} if $\eta\eps^3/2
> 8\,C_1^{-1/2}$.
\end{proof}

\subsection{Serial strategies}\label{sec:directprod}

The main result of this section is Lemma~\ref{lem:productgen}, which shows that serial strategies have a product structure. 
Given that most of the strategies that we consider in this section will have a fixed $q_R$ and $a_R$, we introduce the useful notation $Y_{q_S}^{a_S} := X_{q_R q_S}^{a_R a_S}$ (resp. $\hat{Y}_{q_S}^{a_S} := \hat{X}_{q_R q_S}^{a_R a_S}$) for any $S\subseteq [\ell]\backslash R$; the value of $q_R$ and $a_R$ should always be clear from the context. We will also simply write $Y$ for $X_{q_R}^{a_R}$ (resp. $\hat{Y}$ for $\hat{X}_{q_R}^{a_R}$). For the totality of this section $\eta>0$ is a fixed parameter, which one can think of as polynomial in the soundness $\delta$ that we are aiming for in the repeated game. 

We start with a simple fact which expands on the defining property of $(1-\eta)$-serial strategies.

\begin{fact}\label{fact:serial} Let $q_R\in Q^R$. For every $a_R\in A^R$ there exists $\alpha_{a_R}\geq \Tr\big(X_{q_R}^{a_R}\rho\big)$ such that $\sum_{a_R} \alpha_{a_R} \leq 3$ and the following holds. Suppose $(R,q_R,a_R)$ is $(1-\eta)$-serial, and assume that $\eta \geq C_2^{-1/2}$. Then for a fraction at least $(1-\eta^{1/4})$ of all $(i,q_i)$ for $i\notin R$ we have that
\begin{align}
0 \, \leq\,\Trho\big( \hat{Y}_{q_i}^\dagger\,\hat{Y}_{q_i}\hat{Y}_{q_i}^\dagger\, \hat{Y}_{q_i}\big) -  \sum_{a_i} \Trho\big( (\hat{Y}_{q_i}^{a_i})^\dagger\,\hat{Y}_{q_i}^{a_i}(\hat{Y}_{q_i}^{a_i})^\dagger\, \hat{Y}_{q_i}^{a_i}\big) &\leq 4\,\eta^{3/4}\,\alpha_{a_R} \label{eq:ortgen1}\\
0 \,\leq\,\Tr\big( Y_{q_i}\,\rho^{1/2}Y_{q_i}\,\rho^{1/2}\big) - \sum_{a_i} \Tr\big( Y_{q_i}^{a_i}\, \rho^{1/2} Y_{q_i }^{a_i}\,\rho^{1/2}\big)&\leq 4\,\eta^{3/4}\,\alpha_{a_R}\label{eq:ortgen4}
\end{align}
\end{fact}

\begin{proof}
By condition~\eqref{eq:serial} in the definition of $(1-\eta)$-serial, the $Y_{q_i}^{a_i}$ satisfy
\begin{align}
&E_{(i,q_i)}\Big[  \Trho\big( \hat{Y}^\dagger\,\hat{Y}\hat{Y}^\dagger\, \hat{Y}\big) -  \sum_{a_i} \Trho\big( (\hat{Y}_{q_i}^{a_i})^\dagger\,\hat{Y}_{q_i}^{a_i}(\hat{Y}_{q_i}^{a_i})^\dagger\, \hat{Y}_{q_i}^{a_i}\big) \Big] \notag\\
&\hskip2cm +\,E_{(i,q_i)}\Big[\Tr\big( Y\,\rho^{1/2}Y\,\rho^{1/2}\big) - \sum_{a_i} \Tr\big( Y_{q_i}^{a_i}\, \rho^{1/2} Y_{q_i }^{a_i}\,\rho^{1/2}\big)\Big]\notag\\
&\leq \eta\Big(\Trho\big( \hat{Y}^\dagger\,\hat{Y}\hat{Y}^\dagger\, \hat{Y}\big) +\Tr\big( Y\,\rho^{1/2}Y\,\rho^{1/2}\big)\Big)\label{eq:fs1}
\end{align}
For any $a'_R\in A^R$, let
\beq\label{eq:alphar}
\alpha_{a'_R}\,:=\,\,\max\Big(\Tr(X_{q_R}^{a'_R}\rho), \eta^{-1}\,\Exs{(i,q_i)}{ \big| \Trho\big( (\hat{X}_{q_R}^{a'_R})^\dagger X_{q_R}^{a'_R} \hat{X}_{q_R}^{a'_R}\big) - \Trho\big( (\hat{X}_{q_R q_i}^{a'_R})^\dagger X_{q_R q_i}^{a'_R} \hat{X}_{q_R q_i}^{a'_R}\big) \big|}\Big)
\eeq
By applying Claim~\ref{claim:exptrace2} to the $\hat{X}_{q_R q}^{a'_R}$ we obtain
\begin{align*}
\sum_{a'_R}\Exs{(i,q_i)}{\big| \Trho\big( (\hat{X}_{q_R}^{a'_R})^\dagger X_{q_R}^{a'_R} \hat{X}_{q_R}^{a'_R}\big) - \Trho\big( (\hat{X}_{q_R q_i}^{a'_R})^\dagger X_{q_R q_i}^{a'_R} \hat{X}_{q_R q_i}^{a'_R}\big) \big|} \,\leq\,2 C_2^{-1/2} \Tr(\rho)
\end{align*}
which, by using our assumption that $C_2^{-1/2} \leq \eta$ and $\sum_{a'_R} \Tr(X_{q_R}^{a'_R}\rho)\leq \Tr(\rho)$, implies $\sum_{a'_R} \alpha_{a'_R} \leq 3\Tr(\rho)\leq 3$. Applying Claim~\ref{claim:exptrace3} to the $Y_q^a$ we also obtain
$$ \Exs{(i,q_i)}{\big| \Tr\big(Y\,\rho^{1/2}Y\,\rho^{1/2}\big)- \Tr\big( Y_{q_i }\,\rho^{1/2}Y_{q_i }\,\rho^{1/2}\big) \big|} \leq \eta\, \alpha_{a_R} $$
Hence~\eqref{eq:fs1}, together with an application of Markov's inequality, implies that, for a fraction at least $(1-\eta^{1/4})$ of all $(i,q_i)$,
\begin{align*}
&\Big(\Trho\big( \hat{Y}_{q_i}^\dagger\,\hat{Y}_{q_i}\hat{Y}_{q_i}^\dagger\, \hat{Y}_{q_i}\big) -  \sum_{a_i} \Trho\big( (\hat{Y}_{q_i}^{a_i})^\dagger\,\hat{Y}_{q_i}^{a_i}(\hat{Y}_{q_i}^{a_i})^\dagger\, \hat{Y}_{q_i}^{a_i}\big)\Big)+\Big(\Tr\big( Y_{q_i}\,\rho^{1/2}Y_{q_i}\,\rho^{1/2}\big) - \sum_{a_i} \Tr\big( Y_{q_i}^{a_i}\, \rho^{1/2} Y_{q_i }^{a_i}\,\rho^{1/2}\big)\Big)\\
&\hskip1cm\leq \eta^{3/4}\,\Big(\Trho\big( \hat{Y}^\dagger\,\hat{Y}\hat{Y}^\dagger\, \hat{Y}\big) +\Tr\big( Y\,\rho^{1/2}Y\,\rho^{1/2}\big) + 2 \,\alpha_{a_R}\Big)
\end{align*}
By expanding out the $Y_{q_i}$ terms, one can verify that both terms on the left-hand-side of this equation are positive, hence each of them must be smaller than the right-hand-side, itself smaller than $4\,\eta^{3/4}\alpha_{a_R}$. This proves both~\eqref{eq:ortgen1} and~\eqref{eq:ortgen4}. 
\end{proof}

We now prove a simple claim which shows that $(1-\eta)$-serial strategies are close to being orthogonal; this is how we will subsequently exploit that property. 

\begin{claim}\label{claim:orthogonalgen0} Let $q_R\in Q^R$. For every $a_R\in A^R$ there exists $\alpha_{a_R}\geq \Tr\big(X_{q_R}^{a_R}\rho\big)$ such that $\sum_{a_R} \alpha_{a_R} \leq 3$ and the following holds. Suppose that $(R,q_R,a_R)$ is $(1-\eta)$-serial. Then for a fraction at least $(1-\eta^{1/4})$ of all $(i,q_i)$ for $i\notin R$,
\beq\label{eq:orthogen}
\sum_{a_i\neq a'_i} \Tr_{\rho_{a_i}}\big((\hat{Y}_{q_i }^{ a_i})^\dagger \hat{Y}_{ q_i}^{ a'_i}\, (\hat{Y}_{ q_i}^{ a'_i})^\dagger\, \hat{Y}_{q_i }^{ a_i}\big)\,\leq\,8\eta^{3/4}\,\alpha_{a_R}
\eeq
where $\rho_{a_i} = \rho^{1/2} Y_{q_i}^{a_i} \rho^{1/2}$.
\end{claim}

\begin{proof} Define $\alpha_{a_R}$ as in~\eqref{eq:alphar}.
Letting $Z_i = \hat{Y}_{q_i}^\dagger (\hat{Y}_{q_i}\hat{Y}_{q_i}^\dagger) \hat{Y}_{q_i} - \sum_{a_i}(\hat{Y}_{q_i }^{ a_i})^\dagger \hat{Y}_{ q_i}^{ a_i}\, (\hat{Y}_{ q_i}^{ a_i})^\dagger\, \hat{Y}_{q_i }^{ a_i}$, Eq.~\eqref{eq:ortgen1} from Fact~\ref{fact:serial} can be re-written (for the $(i,q_i)$ for which it holds) as 
$$\Trho(Z_i) \,\leq\, 4\eta^{3/4} \,\alpha_{a_R}$$
Let $\rho_i:= \sum_{a_i} \rho_{a_i}$, where $\rho_{a_i} = \rho^{1/2} Y_{q_i}^{a_i}  \rho^{1/2}$. Since $\rho_i \leq \rho$ and $Z_i \geq 0$, we get
$$\Tr_{\rho_i}(Z_i) \,\leq\,\Trho(Z_i) \,\leq\, 4\eta^{3/4} \,\alpha_{a_R}$$
and hence, expanding out $Z_i$,
\beq\label{eq:ortgen3}
\sum_{a_i\neq a'_i} \Tr_{\rho_i}\big((\hat{Y}_{q_i }^{ a_i})^\dagger \hat{Y}_{ q_i}^{ a'_i}\, (\hat{Y}_{ q_i}^{ a'_i})^\dagger\, \hat{Y}_{q_i }^{ a_i}\big)\,\leq\,4\eta^{3/4}\,\alpha_{a_R}
\eeq
Finally, we can use~\eqref{eq:ortgen4} to upper-bound
$$\sum_{a_i\neq a''_i, a'_i} \Tr_{\rho_{a''_i}}\big((\hat{Y}_{q_i }^{ a_i})^\dagger \hat{Y}_{ q_i}^{ a'_i}\, (\hat{Y}_{ q_i}^{ a'_i})^\dagger\, \hat{Y}_{q_i }^{ a_i}\big)\,\leq\,4\eta^{3/4}\,\alpha_{a_R}$$
where we used $\sum_{a'_i} \hat{Y}_{q_i}^{a'_i}(\hat{Y}_{q_i}^{a'_i})^\dagger \leq \Id$. Together with~\eqref{eq:ortgen3}, this proves the claim.
\end{proof}

\begin{claim}\label{claim:blocksgen} Let $q_R\in Q^R$. For every $a_R\in A^R$ there exists $\alpha_{a_R}\geq \Tr\big(X_{q_R}^{a_R}\rho\big)$ such that $\sum_{a_R} \alpha_{a_R} \leq 3$ and the following holds. Suppose that $(R,q_R,a_R)$ is $(1-\eta)$-serial, let $1 \leq g \leq C_1/2$ be a fixed parameter, and $(G,q_G)$ chosen at random under the constraint that $G\cap R = \emptyset$ and $|G|=g$. Then with probability at least $(1-\eta^{1/4}-e^{-2g})$ over the
choice of $(G,q_G)$, there is a partition $G=G'\cup G''$, where $g''=|G''| \geq (1-4\eta^{c/4})\,g$,
such that for every $i\in G''$
\begin{align}
\sum_{a_i} \Tr_{\rho_G}\big( (\hat{Y}_{q_G}^{a_i})^\dagger  (\Id-\Pi_{q_i}^{a_i})  \hat{Y}_{q_G}^{a_i} \big)  &\leq O\big(g\,\eta^{1/c_2}\big)\, \alpha_{a_R} \label{eq:blocksgen}
\end{align}
where for $i\in G''$, $\{\Pi_{q_i}^{a_i}\}_{a_i}$ is an orthogonal measurement depending only on $q_R,a_R$ and $q_i$ (it is independent of the particular choice of $(G,q_G)$), $\rho_G = \rho^{1/2} Y_{q_G} \rho^{1/2}$, and $c>0,c_2 \geq 1$ are universal constants.
\end{claim}

\begin{proof}
Since $(q_R,a_R)$ is $(1-\eta)$-serial, we can apply Claim~\ref{claim:orthogonalgen0} to obtain that a fraction $(1-\eta^{1/4})$ of $(i,q_i)$ satisfy 
\beq\label{eq:bll1}
\sum_{a_i\neq a'_i} \Tr_{\rho_{a_i}}\big((\hat{Y}_{q_i }^{ a_i})^\dagger \hat{Y}_{ q_i}^{ a'_i}\, (\hat{Y}_{ q_i}^{ a'_i})^\dagger\, \hat{Y}_{q_i }^{ a_i}\big)\,\leq\,8\eta^{3/4}\,\alpha_{a_R}
\eeq
where as before $\rho_{a_i} = \rho^{1/2} Y_{q_i}^{a_i} \rho^{1/2}$. We can now apply Lemma~\ref{lem:blockdiaggen} to the $Y_{q_i}^{a_i}$ (with the states $\rho_{a_i}$) to obtain, for the fraction $(1-\eta^{1/4})$ of $(i,q_i)$ considered above, orthogonal projectors $\{\Pi_{q_i}^{a_i}\}_{a_i}$ satisfying
\begin{align}
\sum_{a_i} \Tr_{\rho_{a_i}}\big((\hat{Y}_{ q_i}^{ a_i })^\dagger (\Id - \Pi_{q_i}^{a_i}) \hat{Y}_{ q_i}^{ a_i }\big)&\leq O\big(\eta^{3c/4}\big)\alpha_{a_R}^c \Big(\sum_{a_i}\Tr\big(\rho_{a_i}\big)\Big)^{1-c}\label{eq:blgen1}
\end{align}
Moreover, the $\Pi_{q_i}^{a_i}$ can easily be made into a projective measurement by enlarging one of them, so that they sum to identity; this will not harm the above bound. 
By Markov's inequality, with probability at least $(1-\eta^{c/4})$ over the choice of $(i,q_i)$ it holds that $\Tr\big(Y_{q_i} \rho\big) \leq \eta^{-c/4} \Tr\big(Y\rho\big) \leq \eta^{-c/4}\alpha_{a_R}$. For any given $(G,q_G)$, let $G''\subseteq G$ denote those indices $i$ in $G$ for which this property holds for $(i,q_i)$, and moreover $(i,q_i)$ falls in the set of indices for which~\eqref{eq:blgen1} holds. By the union bound and a Chernoff bound, the probability that $|G''| \leq (1-4\eta^{c/4})g$ is less than $e^{-2g}$, and for ever $i\in G''$ we have
\begin{align}
\sum_{a_i} \Tr_{\rho_{a_i}}\big((\hat{Y}_{ q_i}^{ a_i })^\dagger (\Id - \Pi_{q_i}^{a_i}) \hat{Y}_{ q_i}^{ a_i }\big)&\leq O\big(\eta^{1/c_2}\big)\alpha_{a_R}\label{eq:blgen1a}
\end{align}
for some constant $c_2>0$. 
Applying Claim~\ref{claim:exptrace3} to the $\hat{Y}_{q_i}^{a_i}$, and summing over $a_i$, we find that in expectation
$$E_{(G,q_G)}\Big[ \sum_{a_i}\big|\Tr_{\rho_{a_i}} \big( (\hat{Y}_{q_i}^{a_i})^\dagger \hat{Y}_{q_i}^{a_i} \big) - \Tr_{\rho_{G,a_i}}\big( (\hat{Y}_{q_G}^{a_i})^\dagger \hat{Y}_{q_G}^{a_i}\big) \big|\Big] \leq g \,C_2^{-1}\,\Tr\big(Y_{q_i} \rho\big) \,\leq\, g \eta^{3/4} \alpha_{a_R}$$
where we used $C_2^{-1} \leq \eta$, $\rho_{G,a_i} := \rho^{1/2} Y_{q_G}^{a_i} \rho^{1/2}$, and we think of the choice of $(G,q_G)$ as first picking $(i,q_i)$ and then the remaining positions and questions. Another application of Claim~\ref{claim:exptrace3} combined with~\eqref{eq:ortgen4} shows that for every $i\in G''$,
$$E_{(G,q_G)}\Big[ \sum_{a_i \neq a'_i}\Tr_{\rho_{G,a_i'}} \big( (\hat{Y}_{q_i}^{a_i})^\dagger \hat{Y}_{q_i}^{a_i} \big)\Big] \leq O(g \,\eta^{3/4})\,\alpha_{a_R}$$
Hence, letting $\rho_{G} := \rho^{1/2} Y_{q_G}\rho^{1/2} = \sum_{a_i} \rho_{G,a_i}$, combining the two previous equations we get 
$$E_{(G,q_G)}\Big[ \sum_{a_i}\big|\Tr_{\rho_{a_i}} \big( (\hat{Y}_{q_i}^{a_i})^\dagger \hat{Y}_{q_i}^{a_i} \big) - \Tr_{\rho_G}\big( (\hat{Y}_{q_G}^{a_i})^\dagger \hat{Y}_{q_G}^{a_i}\big) \big|\Big] \leq O(g \,\eta^{3/4})\,\alpha_{a_R}$$
Using Markov's inequality, his lets us replace $\hat{Y}_{q_i}^{a_i}$ by $\hat{Y}_{q_G}^{a_i}$ in~\eqref{eq:blgen1} for a fraction $(1-\eta^{1/4})$ of $(G,q_G)$, losing an additional factor $O(g\eta^{1/2})\alpha_{a_R}$. Hence
\beq\label{eq:blocks2} \sum_{a_i} \Tr_{\rho_G}\big( (\hat{Y}_{q_G}^{a_i})^\dagger (\Id-\Pi_{q_i}^{a_i})  \hat{Y}_{q_G}^{a_i}\big) \leq  O\big(g\,\eta^{1/c_2}\big)\, \alpha_{a_R}
\eeq
where we safely assumed that $c_2 \geq 2$. 
\end{proof}

\begin{lemma}\label{lem:productgen} Let $q_R\in Q^R$. For every $a_R\in A^R$ there exists $\alpha_{a_R}\geq \Tr\big(X_{q_R}^{a_R}\rho\big)$ such that $\sum_{a_R} \alpha_{a_R} \leq 3$ and the following holds. Under the same conditions as in Claim~\ref{claim:blocksgen}, except for a lower fraction $(1-2\eta^{1/4c_2}-e^{-2g})$ of $(G,q_G)$, it holds that
\begin{align}
\sum_{a_{G''}}\Tr_{\rho_G}\big((\hat{Y}_{q_{G} }^{ a_{G''}} - \hat{Y}_{q_G})^\dagger \Pi_{q_{g''}}^{a_{g''}}\cdots \Pi_{q_1}^{a_1}\cdots \Pi_{q_{g''}}^{a_{g''}} (\hat{Y}_{q_G}^{a_{G''}} - \hat{Y}_{q_G})\big) &\leq O\big( g^2\eta^{1/(4c_2)}\big)\,\alpha_{a_R} \label{eq:prodc2}\\
\sum_{a_{G''}} \Tr_{\rho_G}\big( (\hat{Y}_{q_{G} }^{ a_{G''}})^\dagger \hat{Y}_{q_{G} }^{ a_{G''}}-    (\hat{Y}_{q_G}^{a_{G''}})^\dagger \Pi_{q_{g''}}^{a_{g''}}\cdots \Pi_{q_1}^{a_1}\cdots \Pi_{q_{g''}}^{a_{g''}} \hat{Y}_{q_G}^{a_{G''}} \big) &\leq O\big( g\eta^{1/(8c_2)}\big)\,\alpha_{a_R} \label{eq:prodc1}
\end{align}
\end{lemma}

\begin{proof} Let $\{\Pi_{q_i}^{a_i}\}$ be the orthogonal projectors promised by Claim~\ref{claim:blocksgen}. Let $g'' = |G''|$, and assume for simplicity that the first $g''$ questions in $G$ are those in $G''$. To prove the first inequality, we show the following by induction on $i=1,\ldots,g''$: there exists a constant $C>0$ such that, if we let $F_i = \{1,\ldots,i\}$, then
\beq\label{eq:indblocks}
\sum_{a_{F_i}}\Tr_{\rho_G}\big((\hat{Y}_{q_{G} }^{ a_{F_i}} - \hat{Y}_{q_G})^\dagger \Pi_{q_{i}}^{a_{i}}\cdots \Pi_{q_1}^{a_1}\cdots \Pi_{q_{i}}^{a_{i}} (\hat{Y}_{q_G}^{a_{F_i}} - \hat{Y}_{q_G})\big) \,\leq\, C\, i\,g\,\eta^{1/(3c_2)}\,\alpha_{a_R}
\eeq
The statement for $i=g''$ will imply~\eqref{eq:prodc2}. Let $C_0$ be the constant implicit in~\eqref{eq:blocksgen} from Claim~\ref{claim:blocksgen}. For $i=1$,~\eqref{eq:indblocks} is simply a re-statement of~\eqref{eq:blocksgen}, provided $C$ is chosen larger than $C_0$. Assume the inequality verified for $i-1$, and prove it for $i$. Write 
$$\hat{Y}_{q_G} - \hat{Y}_{q_G}^{a_{F_i}} = (\hat{Y}_{q_G} - \hat{Y}_{q_G}^{a_i}) + (\hat{Y}_{q_G}^{a_{i}} - \hat{Y}_{q_G}^{a_{F_i}})$$
The first term on the right-hand side (when plugged back into~\eqref{eq:indblocks}) can be bounded directly using~\eqref{eq:blocksgen} (and the fact that the projectors $\Pi_{q_j}^{a_j}$ sum to identity over $a_j$, for $j\in\{1,\ldots,i-1\}$). Regarding the second, we can use the Cauchy-Schwarz inequality together with~\eqref{eq:blocksgen} to bound
\begin{align*}
\sum_{a_{F_i}}\big|\Tr_{\rho_G}\big((\hat{Y}_{q_{G} }^{ a_{i}})^\dagger(\Id -  \Pi_{q_{i}}^{a_{i}})\Pi_{q_{i-1}}^{a_{i-1}}\cdots \Pi_{q_1}^{a_1}\cdots \Pi_{q_{i}}^{a_{i}} (\hat{Y}_{q_G}^{a_{F_i}} - \hat{Y}_{q_{G} }^{ a_{i}})\big)\big| \,\leq\, 2\sqrt{C_0} \sqrt{g} \eta^{1/(2c_2)} \alpha_{a_R}^{1/2}\Tr\big(Y_{q_G}\rho\big)^{1/2}
\end{align*}
By Markov's inequality, $\Tr\big(Y_{q_G}\rho\big) \leq \eta^{-1/4c_2}\Tr\big(Y\rho\big)$ for a fraction at least $(1-\eta^{1/4c_2})$ of $(G,q_G)$, so that for those indices the bound above can be replaced by $2\sqrt{C_0} \sqrt{g} \eta^{1/(4c_2)} \alpha_{a_R}$. For the rest of this proof we only consider questions $(G,q_G)$ for which the bound $\Tr\big(Y_{q_G}\rho\big) \leq \eta^{-1/4c_2}\Tr\big(Y\rho\big)$ applies. We can similarly obtain
\begin{align*}
\sum_{a_{F_i}}\big|\Tr_{\rho_G}\big((\hat{Y}_{q_{G} }^{ a_{F_i}})^\dagger(\Id -  \Pi_{q_{i}}^{a_{i}})\Pi_{q_{i-1}}^{a_{i-1}}\cdots \Pi_{q_1}^{a_1}\cdots \Pi_{q_{i}}^{a_{i}} (\hat{Y}_{q_G}^{a_{F_i}} - \hat{Y}_{q_{G} }^{ a_{i}})\big)\big| \,\leq\, 2\sqrt{C_0} \sqrt{g} \eta^{1/(4c_2)} \alpha_{a_R}
\end{align*}
so that
\begin{align*}
\sum_{a_{F_i}}&\Tr_{\rho_G}\big((\hat{Y}_{q_{G} }^{ a_{F_i}} - \hat{Y}_{q_G}^{a_i})^\dagger \Pi_{q_{i}}^{a_{i}}\cdots \Pi_{q_1}^{a_1}\cdots \Pi_{q_{i}}^{a_{i}} (\hat{Y}_{q_G}^{a_{F_i}} - \hat{Y}_{q_G}^{a_i})\big)\\
&\leq \sum_{a_{F_i}}\Tr_{\rho_G}\big((\hat{Y}_{q_{G} }^{ a_{F_i}} - \hat{Y}_{q_G}^{a_i})^\dagger \Pi_{q_{i-1}}^{a_{i-1}}\cdots \Pi_{q_1}^{a_1}\cdots \Pi_{q_{i-1}}^{a_{i-1}} (\hat{Y}_{q_G}^{a_{F_i}} - \hat{Y}_{q_G}^{a_i})\big) + 16\sqrt{C_0}\sqrt{g} \eta^{1/(4c_2)} \alpha_{a_R}\\
& = \sum_{a_{F_i}}\Tr_{\rho_G}\big((\hat{Y}_{q_{G} }^{ a_{F_{i-1}}} - \hat{Y}_{q_G})^\dagger \Pi_{q_{i-1}}^{a_{i-1}}\cdots \Pi_{q_1}^{a_1}\cdots \Pi_{q_{i-1}}^{a_{i-1}} (\hat{Y}_{q_G}^{a_{F_{i-1}}} - \hat{Y}_{q_G}^)\big) + 16\sqrt{C_0}\sqrt{g} \eta^{1/(4c_2)} \alpha_{a_R}
\end{align*}
which can then be bounded using the induction hypothesis. This concludes the induction step, provided $C \geq C_0 + 16\sqrt{C_0}$, and proves~\eqref{eq:prodc2}.   

We now prove~\eqref{eq:prodc1}. Use the Cauchy-Schwarz inequality to bound
\begin{align*}
\sum_{a_{G''}}& \big| \Tr_{\rho_G}\big( (\hat{Y}_{q_{G} }^{ a_{G''}} - \hat{Y}_{q_G})^\dagger \Pi_{q_{g''}}^{a_{g''}}\cdots \Pi_{q_1}^{a_1}\cdots \Pi_{q_{g''}}^{a_{g''}} \hat{Y}_{q_G}^{a_{G''}} \big) \big|\\
&\leq \Big( \sum_{a_{G''}} \Tr_{\rho_G}\big( (\hat{Y}_{q_{G} }^{ a_{G''}} - \hat{Y}_{q_G})^\dagger \Pi_{q_{g''}}^{a_{g''}}\cdots \Pi_{q_1}^{a_1}\cdots \Pi_{q_{g''}}^{a_{g''}} (\hat{Y}_{q_{G} }^{ a_{G''}} -\hat{Y}_{q_G}) \big)\Big)^{1/2}\\
&\hskip.5cm\cdot \Big(\sum_{a_{G''}} \Tr_{\rho_G}\big((\hat{Y}_{q_{G} }^{ a_{G''}} )^\dagger \Pi_{q_{g''}}^{a_{g''}}\cdots \Pi_{q_1}^{a_1}\cdots \Pi_{q_{g''}}^{a_{g''}} \hat{Y}_{q_{G} }^{ a_{G''}}\big)\Big)^{1/2}\\
&\leq O\big(g\eta^{1/(8c_2)}\big)\,\alpha_{a_R}
\end{align*}
by~\eqref{eq:prodc2}. We obtain~\eqref{eq:prodc1} by noting that 
$$\sum_{a_{G''}} \Tr_{\rho_G}\big( (\hat{Y}_{q_{G} }^{ a_{G''}})^\dagger \hat{Y}_{q_{G} }^{ a_{G''}}-    \hat{Y}_{q_G}^\dagger \Pi_{q_{g''}}^{a_{g''}}\cdots \Pi_{q_1}^{a_1}\cdots \Pi_{q_{g''}}^{a_{g''}} \hat{Y}_{q_G} \big) \, = \, 0 $$
since the $\Pi_{q_i}^{a_i}$ sum to identity over $a_i$. 
\end{proof}

\subsection{Bounding the success of players in a repeated game}\label{sec:succbound}

We proceed to show how the results from the two previous sections can be combined in order to prove Theorems~\ref{thm:mainproj} and~\ref{thm:maincons}.  For the remainder of this section we fix a game $G$ with question set $Q$ and answer set $A$, and consider the $\ell$-repeated games $G_{FK(\ell)}$ and $G_{DR(\ell)}$ for some fixed integer $\ell$. Let $s$ be the entangled value of the original game $G$, and $\{A_{q'}^{a'}\}_{a'}$ (resp. $\{(B_q^a)^T\}_{a}$) be an arbitrary fixed projective strategy for Alice (resp. Bob), using entangled state $\ket{\Psi}$, in the $\ell$-repeated game.\footnote{The transpose sign on Bob's operators is there for consistency of notation. For simplicity we will omit this transpose in the future whenever we consider expressions of the form $\bra{\Psi} A\otimes B \ket{\Psi}$, which should be read as $\bra{\Psi} A\otimes B^T \ket{\Psi}$. } Let $\rho = \Tr_A \ket{\Psi}\bra{\Psi}$ be the reduced density of $\ket{\Psi}$ on Bob's subsystem. 

We note here that both types of $\ell$-repeated games have the same overall structure, in that they consist of a set of $C_1$
``correlated'' rounds, in which the referee sends either ``game'' or ``consistency'' questions, and $C_2$ ``independent''
rounds, in which he asks questions chosen independently from a product distribution (we refer to
Definitions~\ref{def:fkrep} and~\ref{def:drrep} for more details, including the definition of $C_1$ and $C_2$). In both cases, we can think of the referee as choosing the $\ell$ pairs of questions in the following order.
\begin{enumerate}
\item First, a subset $R\subseteq [\ell]$ of size $r^*\leq C_1/2$ is chosen, and designated as indices for either game rounds (in the case of a projection game), or otherwise consistency rounds. Pairs of questions $(q'_R,q_R)$ are then picked according to the appropriate distribution.
\item A subset $G\subseteq[\ell]\backslash R$ of size $C_1-r^*$ is chosen. In the case of a projection game, all the indices in $G$ are designated as game rounds. In the other cases, $C_1/2$ of the indices in $G$ are designated (at random) as game rounds, and the remaining indices are designated as consistency rounds. Pairs of questions $(q'_G,q_G)$ are chosen accordingly. Note that the referee doesn't know the value of $r^*$, but he doesn't need to explicitly distinguish between the game and consistency rounds, since they use the same distribution on pairs of questions. The distinction is made only as a convenience for the analysis. 
\item Finally, we let $F = [\ell]\backslash(R\cup G)$. $F$ has size $C_2$, and the indices it contains are designated as confuse rounds, with corresponding pairs of questions $(q'_F,q_F)$.
\end{enumerate}
We will denote by $(q',q):=(q'_Rq'_Gq'_F,q_Rq_Gq_F)$ the $\ell$-tuple of pairs of questions chosen by the referee. Since questions on the indices in $R$ always correspond to cases where for every answer of Alice there is a unique possible valid answer for Bob, and since we will only perform consistency (as opposed to game) checks on questions in those indices, we may regroup Alice's tuples of answers $a'_R$ when they induce the same $a_R$ for Bob. Hence we re-define $A_{q_R q}^{a_R a} := \sum_{a'_R} A_{q_R q}^{a'_R a}$, where the summation runs over all $a'_R$ such that $(a'_R,a_R)$ are valid answers to the questions $(q'_R,q_R)$.

Our first claim shows that the players have a low success probability on blocks $(R,q_R)$ which are dead. 

\begin{claim}\label{claim:deadsuccess} Let $\eps>0$ be such that $\eps \geq C_1 C_2^{-1}$, and suppose that $(R,q_R)$ is an $\eps$-dead block. Then the success probability of the players, conditioned on the referee picking questions $(q',q)$ such that $q$ includes $q_R$ in the positions in $R$, is at most $\sqrt{2\,\eps}$.
\end{claim}

\begin{proof}
The definition of $(R,q_R)$ being $\eps$-dead implies that
$$ \sum_{a_R} \Tr\big(B_{q_R }^{a_R }\,\rho^{1/2}B_{q_R }^{a_R }\,\rho^{1/2}\big)\,\leq\, \eps$$
By applying Claim~\ref{claim:exptrace3} to the $B_{q_R q}^{a_R}$ together with Markov's inequality, we obtain that in expectation
\beq\label{eq:ds1}
 E_{G,q_G}\Big[\sum_{a_R} \Tr\big(B_{q_R q_G}^{a_R }\,\rho^{1/2}B_{q_R q_G}^{a_R }\,\rho^{1/2}\big)\Big] \,\leq\, \eps + C_1 C_2^{-1} \,\leq\, 2\eps
\eeq
where we used $|G|\leq C_1$ and our assumption on $\eps$. Condition on $(q'_R,q_R)$ being chosen as part of the referee's questions in the game, and assume that the referee only checks consistency of Alice and Bob's answers to the questions in $R$. This can only increase their success probability, which can then be bounded as
\begin{align*}
E_{(G,F),(q'_Gq'_F,q_Gq_F)}\Big[\sum_{a_R,a',a}  \langle \Psi | A_{q'_R q'_G q'_F}^{a_R a'} \otimes B_{q_R q_G q_F}^{a_R a}| \Psi \rangle\Big]
& \leq E_{G,(q'_G,q_G)}\Big[ \Big(\sum_{a_R} \| A_{q'_R q'_G}^{a_R}\|_\rho^2\Big)^{1/2}\Big(\sum_{a_R}\|B_{q_R q_G}^{a_R }\|_\rho^2\Big)^{1/2}\Big]\\
& \leq \sqrt{2\,\eps}
\end{align*}
where we used that $(q'_F,q_F)$ are chosen according to a product distribution, the first inequality follows from Cauchy-Schwarz (recall the definition of $\|\cdot\|_\rho$ given in~\eqref{eq:seminorm}), and for the second we upper-bounded $\sum_{a_R}\| A_{q'_R q'_G}^{a_R}\|_\rho^2$ by $1$ and used Jensen's inequality together with~\eqref{eq:ds1} to bound the other term. 
\end{proof}

We note informally that one can combine this claim with Lemma~\ref{lem:productgen} to obtain a form of ``direct product test'' for entangled strategies. Indeed, if two entangled players Alice and Bob win the game with probability $s \gg \eps$, then by the previous claim a fraction at least $s^2/2$ of blocks $(R,q_R)$ should be alive; moreover a non-negligible fraction\footnote{Note that one cannot hope to obtain any structural result on the strategies which would hold for more than a fraction $s$ of questions or answers, as the player's strategy could be a mixture of a perfect winning strategy with probability $s$, and a random strategy with probability $(1-s)$.} of answers $a_R$ to those blocks must be $(1-\eta)$-serial. Hence one can apply Lemma~\ref{lem:productgen} to those blocks $(R,q_R,a_R)$ and obtain a product form for the corresponding marginalized strategy.

The next claim shows that strategies which are product, even on a subset of the coordinates, also have a low success probability.

\begin{claim}\label{claim:prodsuccess} Fix $(R,q_R,a_R)$, and for every $(i,q_i)$, where $i\in [\ell]\backslash R$ and $q_i\in Q$, let $\{\Pi_{q_i}^{a}\}_{a\in A}$ be a fixed projective measurement. Suppose that Bob's strategy is such that, with probability at least $1-\delta$ over the choice of $(G,q_G)$ and $G_1\subseteq G$ of size $|G_1|=g$, there is a partition $G_1 = G'\cup G''$ such that $g''=|G''| \geq (1-\delta')g$ and Bob's POVM satisfies that for every $a_{G''}$
$$ B_{q_R q_G}^{a_R a_{G''}} =  (\hat{B}_{q_R }^{a_R })^\dagger \Pi_{q_{g''}}^{a_{g''}}\cdots \Pi_{q_{1}}^{a_{1}}\cdots \Pi_{q_{g''}}^{a_{g''}}\hat{B}_{q_R }^{a_R }$$
 where for simplicity we wrote $G''=\{1,\ldots,g''\}$.

Then the success probability of the players, conditioned on the referee asking questions $(q',q)$ such that $q$ includes $q_R$ in the positions in $R$, and summed over all valid answers which include $a_R$ for Bob, is at most 
$$\big(\delta+ e^{-(1-s-\delta')^2 g}\big)\, \Tr\big( B_{q_R}^{a_R}\rho\big)$$
\end{claim}

\begin{proof}
Fixing the questions in $R$ and $G$, and conditioning on the players consistently answering $a_R$ to $(q'_R,q_R)$, their probability of being accepted is at most  
\begin{align}
\sum_{a'_{G''},a_{G''}} \bra{\Psi}A_{q'_R q'_G}^{a_R a'_{G''}} \otimes B_{q_R q_G}^{a_R a_{G''}}\ket{\Psi}&= \sum_{a'_{G''},a_{G''}} \bra{\Psi} A_{q'_R q'_G}^{a_R a'_{G''}} \otimes  (\hat{B}_{q_R }^{a_R })^\dagger \Pi_{q_{g''}}^{a_{g''}}\cdots \Pi_{q_{1}}^{a_{1}}\cdots \Pi_{q_{g''}}^{a_{g''}}\hat{B}_{q_R }^{a_R }\ket{\Psi}\notag\\
& = \sum_{a'_{G''},a_{G''}} \big(\bra{\Psi} \Id\otimes (\hat{B}_{q_R}^{a_R})^\dagger \big) \cdot A_{q'_R q'_{G}}^{a_R a'_{G''}} \otimes\Pi_{q_{g''}}^{a_{g''}}\cdots \Pi_{q_{1}}^{a_{1}} \cdots \Pi_{q_{g''}}^{a_{g''}} \cdot \big( \Id\otimes \hat{B}_{q_R}^{a_R} \ket{\Psi}\big)\label{eq:prod2}
\end{align}
The fact that sequential strategies cannot succeed in many rounds of the repeated game implies that
$$ \Big\|E_{(G,q'_G,q_G)}\Big[\sum_{a'_{G''},a_{G''}} A_{q'_R q'_{G}}^{a_R a'_{G''}} \otimes\Pi_{q_{g''}}^{a_{g''}}\cdots \Pi_{q_{1}}^{a_{1}} \cdots \Pi_{q_{g''}}^{a_{g''}}\Big] \Big\|_\infty \,\leq \,\text{exp}(-(1-s-\delta')^2 g) $$
Indeed, the expression on the left-hand side can be upper-bounded by the maximum success probability of an Alice playing an arbitrary strategy and Bob a
sequential strategy described by the measurements $\Pi_{q_i}^{a_i}$, provided the referee only checks the answers
to those questions in $G'' \subseteq G_1$, where $G_1$ is a random subset of $G$ of size $g$ chosen by the referee. But this success probability is even lower than the success probability that
Alice and Bob would have if Bob played his sequential strategy on \emph{all} questions in $G_1$, but the referee was
to accept as long as at least $g''$ out of Alice and Bob's $g$ answers were correct. Since the probability of such a
serial strategy succeeding in any round is at most the value $s$ of the original game, and $g'' \geq (1-\delta')g$, by a
Chernoff bound the probability that the players succeed in $g''$ out of the $g$ rounds is at most
$\text{exp}(-(1-s-\delta')^2 g)$. Hence the expression in~\eqref{eq:prod2} can be upper-bounded, in expectation, by 
$$e^{-(1-s-\delta')^2 g}\, \bra{\Psi}\Id\otimes (\hat{B}_{q_R}^{a_R})^\dagger \hat{B}_{q_R}^{a_R}\ket{\Psi} \,=\, e^{-(1-s-\delta')^2 g}\, \Tr\big( B_{q_R}^{a_R}\rho\big)$$ 

Finally, we must account for the small probability $\delta$ that the serial property does not hold; for those sets $G$ we can trivially bound the success probability, conditioned on Bob answering $a_R$ to $q_R$, by $\Tr\big(B_{q_R}^{a_R}\rho\big)$.
\end{proof}

We finally turn to the proof of our main theorem. 

\begin{proof}[Proof of Theorem~\ref{thm:mainproj}]
We first set parameters: let $C_0$ be a large enough constant, $\eps = C_0^{-1}\delta^{2}$ (recall that $\delta$ is the target value for the repeated game $G_{FK}(\ell)$), $\eta =  C_0^{-1}\delta^{24c_2}(1-s)$ (where $c_2$ is the constant which appears in Claim~\ref{claim:blocksgen}), $g = C_0 \log(1/\delta)(1-s)^{-1}$, and $\ell \geq C_0^{15} \delta^{-125c_2}(1-s)^{-4}$. Recall also that $C_1$ was defined as $C_1 = \sqrt{\ell}$, and $C_2 = \ell-C_1$. This choice of parameters satisfies the following constraints:
\begin{itemize}
\item $\eta\,\eps^3 > 16 \,C_1^{-1/2}$, which is used in Lemma~\ref{lem:alive}.
\item $\eta \geq C_2^{-1/2}$, which is used in Fact~\ref{fact:serial} and subsequent claims.
\item $\eps \geq C_1\, C_2^{-1}$, which is used in Claim~\ref{claim:deadsuccess}. 
\end{itemize}

As before, in game $G_{FK(\ell)}$, we can think of the referee as first picking $r^* \leq C_1/2$ pairs of questions $(R,(q'_R,q_R))$
for the players, then picking $g$ pairs $(G_1,(q'_{G_1},q_{G_1}))$, then $C_1-r^*-g$ pairs $(G_2,(q'_{G_2},q_{G_2}))$ and finally
$C_2$ independent pairs of confuse questions $(F,(q'_F,q_F))$. Let $G=G_1\cup G_2$ and $(q',q)=(q'_Rq'_Gq'_F,q_Rq_Gq_F)$. Let
$\{A_{q'}^{a'}\}_{a'}$ be Alice's POVM on questions $q'$, and $\{B_q^a\}_{a}$ Bob's POVM on questions $q$.

By Lemma~\ref{lem:alive}, one of two cases hold. Either a $(1-\eps)$ fraction of blocks $(R,q_R)$ are $\eps$-dead, in which case the player's success probability is readily bounded by $\eps+\sqrt{2\eps}$ by Claim~\ref{claim:deadsuccess}. Otherwise, it must be that we are in case 2 of the lemma, so that $\eps$-alive blocks are for the most part serial. Note that any dead blocks contribute at most $\sqrt{2\eps}$ to the success probability, by Claim~\ref{claim:deadsuccess}. A similar argument to that in Claim~\ref{claim:deadsuccess} shows that alive blocks which are not $(1-\eta)$-serial also contribute at most $\sqrt{2\eps}$, given the fact that we are in the case 2. of Lemma~\ref{lem:alive}, and there can only be few such blocks by~\eqref{eq:alivenotserial}.

Suppose $(R,q_R,a_R)$ is $(1-\eta)$-serial. By Lemma~\ref{lem:productgen}, for every $(i,q_i)$ there exists a projective measurement $\{\Pi_{q_i}^{a_i}\}_{a_i}$, depending only on $q_R,a_R,q_i,a_i$, such that with probability at least $(1-2\eta^{1/4c_2}-e^{-2g})$ over the choice of $(G,q_{G})$ such that $|G|=g$ there is a partition $G_1 = G'\cup G''$ such that $g''=|G''| \geq (1-4\eta^{c/4})g$ such that Eqs.~\eqref{eq:prodc2} and~\eqref{eq:prodc1} from Lemma~\ref{lem:productgen} are satisfied, 
where $\rho_G = \rho^{1/2} B_{q_R q_G}^{a_R} \rho^{1/2}$. To alleviate notation we let $\Pi =  \Pi_{q_1}^{a_1}\cdots \Pi_{q_{g''}}^{a_{g''}}$, and we first use Cauchy-Schwarz to bound
\begin{align}
\sum_{a'_{G''},a_{G''}}  & \bra{\Psi} A_{q'_R q'_G }^{a_R a'_{G''}} \otimes (\hat{B}_{q_R q_G}^{a_R a_{G''}})^\dagger(\Id - \Pi^\dagger\Pi) \hat{B}_{q_R q_G}^{a_R a_{G''}} \ket{\Psi} \notag\\
&\leq  \|A_{q'_R q_G}^{a_R }\|_\rho \, \Big\| \sum_{a_{G''}} ( \hat{B}_{q_R q_G }^{a_R a_{G''}})^\dagger(\Id - \Pi^\dagger \Pi) \hat{B}_{q_R q_G}^{a_R a_{G''}} \Big\|_\rho\notag\\
&\leq  \|A_{q'_R q_G}^{a_R }\|_\rho \, \Big(\sum_{a_{G''}} \Tr_{\rho_G} \big(( \hat{B}_{q_R q_G }^{a_R a_{G''}})^\dagger(\Id - \Pi^\dagger \Pi) \hat{B}_{q_R q_G}^{a_R a_{G''}}\big)\Big)^{1/2}\notag\\
&\leq O\big(\sqrt{g} \eta^{1/(16c_2)}\big) \|A_{q'_R q_G}^{a_R }\|_\rho\, \alpha_{a_R}^{1/2}\label{eq:fbound1}
\end{align}
where $\rho_G = \rho^{1/2}B_{q_R q_G}^{a_R} \rho^{1/2}$, the first inequality is by Cauchy-Schwarz, the second uses $(\Id-\Pi^\dagger\Pi)\leq\Id$, the last is by Eq.~\eqref{eq:prodc1} from Lemma~\ref{lem:productgen}, and $\alpha_{a_R}$ was defined in Eq.~\eqref{eq:alphar} (where here we substitute $\hat{B}_{q_R}^{a_R}$ for $\hat{X}_{q_R}^{a_R}$). A similar argument, using this time Eq.~\eqref{eq:prodc2}, lets us bound 
\begin{align}
\sum_{a'_{G''},a_{G''}}   \bra{\Psi} A_{q'_R q'_G }^{a_R a'_{G''}} \otimes \big(\hat{B}_{q_R q_G}^{a_R a_{G''}} - \hat{B}_{q_R q_G}^{a_R }\big)^\dagger  \Pi^\dagger\Pi \big(\hat{B}_{q_R q_G}^{a_R a_{G''}} - \hat{B}_{q_R q_G}^{a_R }\big) \ket{\Psi}  \leq  O\big(g \eta^{1/(8c_2)}\big) \|A_{q'_R q_G}^{a_R }\|_\rho\, \alpha_{a_R}^{1/2}\label{eq:fbound2}
\end{align}
and hence combining~\eqref{eq:fbound1} and~\eqref{eq:fbound2} we get
\begin{align*}
\sum_{a'_{G''},a_{G''}}  \big| \bra{\Psi} A_{q'_R q'_G }^{a_R a'_{G''}} \otimes \big(B_{q_R q_G}^{a_R a_{G''}} - (\hat{B}_{q_R q_G}^{a_R })^\dagger  \Pi^\dagger\Pi \hat{B}_{q_R q_G}^{a_R} \big) \ket{\Psi} \big| \leq  O\big(\sqrt{g} \eta^{1/(16c_2)}\big) \|A_{q'_R q_G}^{a_R }\|_\rho\, \alpha_{a_R}^{1/2}
\end{align*}
Finally, by Claim~\ref{claim:exptrace3} we have
\begin{align*}
\textsc{E}_{(G,q_G)} &\Big[\sum_{a'_{G''},a_{G''}}\big| \bra{\Psi} A_{q'_R q'_G }^{a_R a'_{G''}} \otimes \big((\hat{B}_{q_R}^{a_R })^\dagger  \Pi^\dagger\Pi \hat{B}_{q_R}^{a_R} - (\hat{B}_{q_R q_G}^{a_R })^\dagger  \Pi^\dagger\Pi \hat{B}_{q_R q_G}^{a_R} \big) \ket{\Psi} \big|\Big]\\
&\leq 4\|A_{q'_R q_G}^{a_R }\|_\rho\,\Exs{(G,q_G)}{ \big| \big\|B_{q_R}^{a_R } \big\|_\rho^2 - \big\|B_{q_R q_G}^{a_R } \big\|_\rho^2}^{1/2}\\
&\leq 4\eta \|A_{q'_R q_G}^{a_R }\|_\rho\, \alpha_{a_R}^{1/2}
\end{align*}
where for the first inequality we used $\sum_{a_G''}\Pi^\dagger\Pi = \Id$, and for the second that $\eta\geq C_2^{-1}$.
Hence the statistical distribution of outcomes produced by Alice and Bob (conditioned on answering $a_R$ to $q_R$) is close to that which would be obtained if Bob was to use the operators $(B_{q_R}^{a_R })^\dagger \Pi^\dagger \Pi B_{q_R}^{a_R }$ as his POVM on questions $q_G$. But the success probability of the latter, when summed over all valid answers to the pair of questions $(q'_{G''},q_{G''})$, can be bounded by Claim~\ref{claim:prodsuccess}.
 Hence summing over all $a_R$ (and using $\sum_{a_R}\|A_{q'_R q_G}^{a_R }\|_\rho\, \alpha_{a_R}^{1/2} \leq 3$)  and taking the expectation over $q_R$, the average winning probability of the players for all $(1-\eta)$-serial blocks $(R,q_R,a_R)$ is at most 
$$O\big(\sqrt{g}\, \eta^{1/(16c_2)} +2\eta^{c/4}+ e^{-2g}+ e^{-(1-s-4\eta^{1/4c_2})^2g}\big)$$
where we also accounted for those (rare) choices of $(G,q'_G,q_G)$ for which the previous bounds do not hold. 
 Given our choice of parameters $\eps,\eta, g$ and $\ell$, it can be checked that this expression is $\ll\delta$. Combining this bound with the one resulting from dead blocks shows that the winning probability of the players is at most $\delta$, which proves the theorem as long as $\ell = \poly(\delta^{-1},(1-s)^{-1})$ is large enough. 
\end{proof}

We conclude this section by briefly explaining how the proof of Theorem~\ref{thm:mainproj} can be adapted to prove Theorem~\ref{thm:maincons}. The main reason the proof carries over is that, in the proof of Theorem~\ref{thm:mainproj}, we only used the projection property for a subset of the game questions (to bound the success over dead blocks), while for $(1-\eta)$-serial blocks the game questions were only used in conjunction with the fact that the value of the game was at most $s$. Here, consistency rounds will play the role of the game questions previously in $R$, and game rounds will play the role of those game questions previously in $G$ (or rather its small subset $G_1$).

\begin{proof}[Proof of Theorem~\ref{thm:maincons}]
In game $G_{DR(\ell)}$, we think of the referee as first picking $r^* \leq C_1/2$ pairs of consistency questions
$(R,(q'_R,q_R))$ for the players, then picking $C_1/2-r^*$ additional consistency pairs $(R',(q'_{R'},q_{R'}))$, $C_1/2$ pairs  of
game questions $(G,(q'_{G},q_{G}))$ and finally $C_2$ independent pairs of confuse questions $(F,(q'_F,q_F))$. Let
$(q',q)=(q'_Rq'_{R'}q'_Gq'_F,q_Rq_{R'}q_Gq_F)$. 

Assume a choice of parameters made that is similar to the one in the proof of Theorem~\ref{thm:mainproj}. As before, we can apply Lemma~\ref{lem:alive} to Bob's strategy $B_q^a$, distinguishing between two cases.

In the first case, a fraction $(1-\eps)$ of blocks $(R,q_R)$ are dead, for $|R|=r^*$. Then Claim~\ref{claim:deadsuccess} again applies, as the only property we used in its proof was that any answer of Alice induced a fixed answer for Bob, which is the case for consistency questions.

In the second case, a fraction $\eps$ of blocks $(R,q_R)$ are alive. Those blocks which are dead can be dealt with as in the previous case, and we can focus on blocks $(R,q_R,a_R)$ which are $(1-\eta)$-serial. Here we can reason exactly as in Theorem~\ref{thm:mainproj}, using Claim~\ref{claim:prodsuccess} with $G_1$ chosen as a subset of the questions in $G$, and the remaining consistency questions playing the role of the remaining game questions before.
\end{proof}


\section{Approximate block-diagonalization of almost-orthogonal operators}\label{sec:blockdiag}

In this section we prove our {\em orthogonalization lemma}, Lemma~\ref{lem:blockdiaggen} below, which shows that pairwise
almost-orthogonal operators are close to having a joint block-diagonal decomposition. The main ingredient in its proof is a robust orthogonalization lemma for families of pairwise almost-orthogonal projectors, Lemma~\ref{lem:kproj}. 

The proof of Lemma~\ref{lem:kproj} is based on a variant of Sch\"oneman's solution to the ``orthogonal Procrustes\footnote{According to Wikipedia, Procrustes, or ``the stretcher'', a figure from Greek mythology, was a rogue smith and bandit from Attica who physically attacked people, stretching them, or cutting off their legs so as to make them fit an iron bed's size.} problem''. Given any square matrices $A$ and $B$, this is the problem of finding the orthogonal matrix $\Omega$ which minimizes
$$\Omega:=\,\text{argmin } \| A-B\Omega\|_F^2 $$
where $\|X\|_F^2 = d^{-1} \Tr(X^\dagger X)$ is the normalized Frobenius norm.  
Sch\"oneman~\cite{Schonemann1966} showed that the optimal $\Omega$ is $\Omega = UV^\dagger$, where $U\Sigma V^\dagger$ is the singular value decomposition of $B^T A$.\footnote{We are grateful to the user ``ohai'' of MathOverflow.net for pointing out the connection between this problem and that of the robust orthonormalization of almost-orthogonal vectors.} Indeed, given unit vectors $\ket{u_1},\ldots,\ket{v_k}$, one can let $A$ be the matrix with columns the $\ket{u_i}$, and $B$ the identity. In this case, the orthogonal Procruste's problem consists in finding the best rigid rotation which maps the canonical basis of space to the vectors $\ket{v_i}$, where the error is measured in the least squares sense --- the columns of the corresponding orthogonal matrix will then form an orthonormal family close to the $\ket{u_i}$. 

We carry out this solution precisely in Claim~\ref{claim:vectors} below, which, even though we will not use it directly, contains all the intuition necessary to solve our original problem on positive matrices. Unfortunately, the solution to the latter is made more involved technically by the the matrices not being of rank $1$, and the slightly unorthodox (and, in particular, not rotationally invariant) way in which we measure the error.  

\begin{claim}\label{claim:vectors} Let $\ket{u_1},\ldots,\ket{u_k} \in\C^k$ be unit vectors such that $\frac{1}{k}\sum_{i\neq j} \langle u_i , u_j \rangle^2 \leq \eps$. Then there exist orthogonal unit vectors $\ket{v_1},\ldots,\ket{v_k} \in \C^k$ such that $\frac{1}{k} \sum_i \big\|\,\ket{u_i}-\ket{v_i}\,\big\|^2 \leq \eps$.
\end{claim}

\begin{proof}
Let $X$ be the $k\times k$ matrix whose columns are made of the vectors $\ket{u_i}$, expressed in the canonical basis. The SVD of $X$ is $X = U \Sigma V^\dagger$, where $U,V$ are unitary and $\Sigma$ is diagonal with the singular values $s_i$ of $M$ on the the diagonal. Then 
\beq\label{eq:sibound}
 \frac{1}{k} \sum_{i=1}^k (1-s_i^2)^2 \,=\, \|\Sigma^\dagger \Sigma -\Id \|_F^2 \, = \,\|X^\dagger X - \Id \|_F^2 \,=\, \frac{1}{k}\sum_{i\neq j} \big|\langle u_i , u_j \rangle\big|^2 \, \leq \,\eps 
 \eeq
where for the first equality we used the unitary invariance of the Frobenius norm, and the second is by definition of $X$ and uses the fact that the $\ket{u_i}$ have unit norm. Let $Y = UV^\dagger$. $Y$ is a unitary matrix so its column vectors $\ket{v_i}$ form an orthonormal family. We have
$$
\frac{1}{k} \sum_{i=1}^k \big\|\,\ket{u_i}-\ket{v_i}\,\big\|_2^2 \,=\, \|X-Y\|_F^2 \, =\, \|\Id -\Sigma\|_F^2 \, = \, \frac{1}{k} \sum_{i=1}^k (1-s_i)^2
$$
which can be bounded by~\eqref{eq:sibound} since $(1-s_i)^2 \leq (1-s_i)^2(1+s_i)^2 = (1-s_i^2)^2$. 
\end{proof}

We now extend this claim to the case of almost-orthogonal projections, which need not have rank $1$, and to a slightly different way of measuring the error (most of the difficulty in proving the lemma comes from the different norm rather than from the higher rank). In order to understand the following, it may be helpful to first consider the case where $\rho_i = (dk)^{-1} \Id$ for every $i$. 

\begin{lemma}\label{lem:kproj} Let $\rho_i$, $i=1,\ldots,k$ be positive matrices, and $\rho := \sum_i \rho_i$. Let $P_1,\ldots,P_k$ be $d$-dimensional projectors such that $$\sum_{i\neq j} \Tr(P_iP_jP_i\, \rho_i) \leq \eps\qquad \text{and}\qquad \sum_{i\neq j} \Tr(P_i\,\rho_j) \leq\eps$$
 for some $0< \eps \leq \Tr(\rho)$. Then there exists orthogonal projectors $Q_1,\ldots,Q_k$ such that 
 $$\sum_{i=1}^k \Tr\big((P_i - Q_i)^2\,\rho_i\big) = O\big(\eps^{1/2}\big)\, \Tr(\rho)^{1/2}$$
\end{lemma}

\begin{proof}
For every $i$ write $P_i = \sum_l \ket{x_{i,l}}\bra{x_{i,l}}$, where the $\{\ket{x_{i,l}}\}_{l}$ are orthonormal, and let $X_i:=\sum_l \ket{x_{i,l}}\bra{e_{i,l}}$, $X:=\sum_i X_i$, where $\ket{e_{i,l}}$ is the canonical basis: $X$ has the $\ket{x_{i,l}}$ as its columns. In order for $X$ to be a square matrix, if necessary we extend the space in which the $\ket{x_{i,l}}$ vectors live, so as to make it the same dimension as $\text{Span}\{\ket{e_{i,l}}\}$. 
The inner-product condition on the $P_i$ implies that 
\beq\label{eq:pppbound1}
\sum_{i\neq j} \Tr\big(P_iP_jP_i\,\rho_i\big) \,=\, \sum_{i\neq j} \sum_{l,l',l''} \bra{x_{i,l}} x_{j,l'} \rangle \bra{x_{j,l'}} x_{i,l''}\rangle \bra{x_{i,l''}} \rho_i \ket{x_{i,l}} \,\leq\,\eps
\eeq
Write $X^\dagger X = \sum_{i,j,l,l'} \bra{x_{i,l}} x_{j,l'} \rangle \,\ket{e_{i,l}}\bra{e_{j,l'}}$, so that 
\beq\label{eq:pppbound2}
\sum_i \Tr \big(\big( X^\dagger X - \Id)^2\,X_i^\dagger \rho_i X_i \big) \,=\, \sum_{i,l,l''} \sum_{(j,l')\neq (i,l),(i,l'')} \bra{x_{i,l}} x_{j,l'} \rangle \bra{x_{j,l'}} x_{i,l''}\rangle \bra{x_{i,l''}} \rho_i \ket{x_{i,l}} \,\leq\,\eps
\eeq
where we used~\eqref{eq:pppbound1} to upper-bound the expression in the middle by $\eps$. Indeed, in the second summation, if $i=j$ then either $l' \neq l$ or $l' \neq l''$, so that one of the inner products $\bra{x_{i,l}} x_{i,l'}\rangle$ or $\bra{x_{i,l'}} x_{i,l''}\rangle$ is $0$, since the $\{\ket{x_{i,l}}\}_l$ are orthogonal. 
 
Let $X = U \Sigma V^\dagger$, where $\Sigma$ is diagonal positive and $U,V$ unitary, be the polar decomposition of $X$. By an appropriate choice of the basis $\ket{e_{i,l}}$ we can assume that $V = \Id$ (if not, re-define $X_i:= X_i V$; this corresponds to changing $\ket{e_{i,l}}\to V^\dagger \ket{e_{i,l}}$). Let $\Pi$ be the projector on the span of the eigenvectors of $\Sigma$ with corresponding eigenvalue at least $1/2$ and at most $2$. $\Pi$ is needed to control eigenvalues of $\Sigma$ which may be too small or too large. 

 Let $\tilde{U} = U\Pi$ and $\tilde{X} = X\Pi$. Let $\ket{\tilde{u}_{i,l}}$ (resp. $\ket{\tilde{x}_{i,l}}$) be the column vectors of $\tilde{U}$ (resp. $\tilde{X}$), so that $\tilde{U} = \sum_{i,l} \ket{\tilde{u}_{i,l}}\bra{e_{i,l}}$. 
We will show that the projectors $Q_i := \sum_l \ket{\tilde{u}_{i,l}}\bra{\tilde{u}_{i,l}}$ are close to the projectors $P_i$, in the sense claimed in the lemma (note that since $U$ is unitary and $\Pi$ a diagonal projector the $Q_i$ are orthogonal projectors, which do not necessarily sum to identity). We first state some consequences of~\eqref{eq:pppbound2}. 

\begin{fact}\label{fact:temp1} The following inequalities holds
\begin{align}
&\sum_{i,l,l'} \bra{\tilde{u}_{i,l}-\tilde{x}_{i,l}}\tilde{u}_{i,l'} - \tilde{x}_{i,l'}\rangle \bra{\tilde{x}_{i,l'}}\rho_i\ket{\tilde{x}_{i,l}}\,\leq\, \eps \label{eq:illbound}\\
&\sum_{i,l} |\bra{\tilde{u}_{i,l}}\rho\ket{\tilde{u}_{i,l}} - \bra{\tilde{x}_{i,l}}\rho\ket{\tilde{x}_{i,l}}| \,\leq\, 2\sqrt{2}\,\eps^{1/2} \Tr(\rho)^{1/2} \label{eq:uxrho}
\end{align}
\end{fact}

\begin{proof}
We start with proving~\eqref{eq:illbound}. Since $\Sigma$ is diagonal, one can immediately check that $X^\dagger X - \Id = (X-U)^\dagger(X+U)$. Note also that $(X+U)(X+U)^\dagger = U(\Id+\Sigma)^2 U^\dagger \geq \Id$. Hence
\begin{align}
\sum_{i} \Tr\big((\Sigma-\Id)^2 X_i^\dagger \rho_i X_i \big) &=
  \sum_i\Tr\big((X-U)^\dagger (X-U) X_i^\dagger \rho_i X_i \big)\notag\\
  &\leq \sum_i\Tr\big((X-U)^\dagger (X+U)(X+U)^\dagger (X-U) X_i^\dagger \rho_i X_i \big) \notag\\
  &\leq \eps \label{eq:ux}
\end{align}
where the last inequality is by~\eqref{eq:pppbound2}. This implies that $\sum_{i} \Tr((\Sigma-\Id)^2 (X_i \Pi)^\dagger \rho_i (X_i\Pi) \leq \eps$ (note that $\Pi$ commutes with $\Sigma$ by definition), which is just~\eqref{eq:illbound}. 

Before turning to the proof of~\eqref{eq:uxrho}, first observe that
\begin{align}
 \Tr((\Sigma-\Id)^2 \Pi X^\dagger \rho X) &= \sum_{i,j} \Tr((\Sigma-\Id)^2 \Pi X_i^\dagger \rho_j X_i\big)\notag\\
 &\leq 2\eps \label{eq:sigmaxbound}
\end{align} 
where the equality uses that $(\Sigma-\Id)^2\Pi)$ is diagonal, and the inequality is by~\eqref{eq:illbound} for the terms $i=j$ and uses $(\Sigma-\Id)^2\Pi \leq \Id$ and the second condition in the lemma for the terms $i\neq j$. From~\eqref{eq:sigmaxbound} we get
\begin{align}
\sum_{i,l} \bra{ \tilde{u}_{i,l} - \tilde{x}_{i,l}} \rho \ket{\tilde{u}_{i,l} - \tilde{x}_{i,l}} &= \Tr( \Pi (X-U)^\dagger \rho (X-U)) \notag\\
& \leq 4\,\Tr\big( \Sigma \Pi \Sigma (X-U)^\dagger \rho (X-U)\big) \notag\\
&= 4\,\Tr\big( (\Id-\Sigma)  \Pi (\Id-\Sigma) X^\dagger \rho X\big) \notag\\
& \leq 8\eps\label{eq:uminvrho}
\end{align}
where the first inequality uses $\Pi\Sigma \geq 1/2 \Pi$, by definition of $\Pi$, and the last is by~\eqref{eq:sigmaxbound}.

We now prove~\eqref{eq:uxrho}. By Cauchy-Schwarz, for every $(i,l)$
$$\bra{\tilde{u}_{i,l}-\tilde{x}_{i,l}}\rho\ket{\tilde{u}_{i,l}} \leq \bra{\tilde{u}_{i,l}-\tilde{x}_{i,l}}\rho\ket{\tilde{u}_{i,l}-\tilde{x}_{i,l}}^{1/2} \bra{\tilde{u}_{i,l}} \rho \ket{\tilde{u}_{i,l}}^{1/2}$$
hence by~\eqref{eq:uminvrho} we see that 
$$\sum_{i,l} |\bra{\tilde{u}_{i,l}-\tilde{x}_{i,l}}\rho\ket{\tilde{u}_{i,l}}| \leq 2\sqrt{2}\,\eps^{1/2} \Tr(\rho)^{1/2}$$
A symmetric inequality can be obtained, and~\eqref{eq:uxrho} follows by the triangle inequality. 
\end{proof} 

As a consequence of Fact~\ref{fact:temp1}, note that
\begin{align}
\Big|\sum_{i,l,l'} \bra{\tilde{u}_{i,l}} \tilde{x}_{i,l'}\rangle \,\bra{\tilde{x}_{i,l'}} \rho_i \ket{\tilde{u}_{i,l}-\tilde{x}_{i,l}}\Big|&\leq \Big(\sum_{i,l,l'} \bra{\tilde{x}_{i,l}}\tilde{x}_{i,l'}\rangle\bra{\tilde{x}_{i,l'}} \rho_i\ket{\tilde{x}_{i,l}} \Big)^{1/2}\Big(\sum_{i,l} \bra{\tilde{u}_{i,l}-\tilde{x}_{i,l}}\rho_i \ket{\tilde{u}_{i,l}-\tilde{x}_{i,l}}\Big)^{1/2}\notag\\
&\leq \Tr(\rho)^{1/2}\cdot ( 8\eps)^{1/2} \,=\, O(\eps^{1/2})\Tr(\rho)^{1/2} \label{eq:uxrho2}
\end{align}
where the first inequality is by Cauchy-Schwarz (and the $\ket{\tilde{u}_{i,l}}$ being orthonormal) and the second uses $\tilde{X}_i\tilde{X}_i^\dagger \leq \Id$, and~\eqref{eq:uminvrho} (with $\rho_i \leq \rho$).

In order to bound the distance between $Q_i = \sum_l \ket{\tilde{u}_{i,l}}\bra{\tilde{u}_{i,l}}$ and $P_i$, we first bound the distance between $Q_i$ and $\tilde{P}_i:= \tilde{X}_i\tilde{X}_i^\dagger$:
\begin{align}
\sum_i \Tr\big((\tilde{P}_i - Q_i)^2\,\rho_i\big) & =  \sum_{i,l} \big(\bra{\tilde{x}_{i,l}} \rho_i \ket{\tilde{x}_{i,l}} + \bra{\tilde{u}_{i,l}} \rho_i \ket{\tilde{u}_{i,l}}\big) - 2\sum_{i,l,l'} \Re\big(\bra{\tilde{u}_{i,l}} \tilde{x}_{i,l'}\rangle \,\bra{\tilde{x}_{i,l'}} \rho_i \ket{\tilde{u}_{i,l}}\big) \notag\\
&\leq 2\sum_{i,l} \bra{\tilde{x}_{i,l}} \rho_i \ket{\tilde{x}_{i,l}} - 2\sum_{i,l,l'} \Re\big(\bra{\tilde{u}_{i,l}} \tilde{x}_{i,l'}\rangle \,\bra{\tilde{x}_{i,l'}} \rho_i \ket{\tilde{x}_{i,l}}\big)+ O(\eps^{1/2} \Tr(\rho)^{1/2})\notag\\
& \leq O(\eps^{1/2} \Tr(\rho)^{1/2})\label{eq:piqi}
\end{align}
where the first inequality is by~\eqref{eq:uxrho} and~\eqref{eq:uxrho2} and the second by~\eqref{eq:illbound}. It remains to bound the distance between the $\tilde{P}_i$ and the $P_i$:
\begin{align}
\sum_i \Tr\big((\tilde{P}_i - P_i)^2\,\rho_i\big) &=  \sum_i \Tr\big( (\Id-\Pi) X_i^\dagger \rho_i X_i \big)\notag\\
&\leq 2\sum_i \Tr\big( |\Id-\Sigma| X_i^\dagger \rho_i X_i \big)\notag\\
&\leq 2\Big(\sum_i \Tr\big( (\Id-\Sigma)^2 X_i^\dagger \rho_i X_i \big)\Big)^{1/2}\Big(\sum_i \Tr\big( X_i^\dagger \rho_i X_i \big)\Big)^{1/2}\notag\\
&\leq 2\eps^{1/2} \Tr(\rho)^{1/2}\label{eq:pipitildez}
\end{align}
where the first inequality uses $(\Id-\Pi) \leq 2|\Sigma - \Id|$ by definition of $\Pi$, the second is Cauchy-Schwarz and the last is by~\eqref{eq:ux}. 
Combining~\eqref{eq:piqi} and~\eqref{eq:pipitildez} finishes the proof of the lemma. 
\end{proof}

Lemma~\ref{lem:kproj} lets us prove the orthogonalization lemma below. In that lemma one can think of the $\hat{Y}_i$ as operators in the Stinespring representation of a measurement $\mathcal{M}_i:\rho\mapsto \hat{Y}_i (\rho\otimes \Id) \hat{Y}_i^\dagger$, where $i$ refers to the $i$-th outcome of the measurement. In that setting the hypothesis of the lemma is that, when $\mathcal{M}$ is performed twice sequentially on a specific state $\rho$, it is likely that identical answers will be obtained. The conclusion is that the operators $\hat{Y}_i$ have an approximate joint block-diagonal form, as described by the orthogonal projectors $\Pi_i$. 

\begin{lemma}\label{lem:blockdiaggen}[Orthogonalization Lemma] There is a $c>0$ such that the following holds. 
Let $\rho_i$, $i=1,\ldots,k$ be positive, $\rho$ such that $\sum_i \rho_i \leq \rho$ and $\hat{Y}_i$, $i=1,\ldots,k$ (possibly rectangular) matrices, be  such that 
\begin{align}
\sum_{i\neq j} \Tr_{\rho_i}\big( \hat{Y}_i^\dagger\, (\hat{Y}_j\,  \hat{Y}_j^\dagger)\, \hat{Y}_i\big) &\leq \alpha \,\Tr(\rho) \label{eq:ortho}
\end{align}
and $\sum_i \hat{Y}_i \hat{Y}_i^\dagger\leq\Id$. Then there exists
orthogonal projectors $\{\Pi_i\}$ such that
$$ \sum_i \Tr_{\rho_i}\big(\hat{Y}_i^\dagger (\Id-\Pi_i) \hat{Y}_i \big) \,\leq O\big( \alpha^{c}\big)\Tr(\rho)$$
\end{lemma}

\begin{proof} The idea of the proof is simple. Let $\beta_1,\beta_2>0$ be parameters to be chosen later. For every $i$, let $P_i$ be the projector on the eigenvectors of $\hat{Y}_i\,  \hat{Y}_i^\dagger$ with corresponding eigenvalue at least $\beta_1$. Since $P_i$ contains all the large eigenvalues, $P_i \hat{Y}_i \approx \hat{Y}_i$. Moreover, by definition $P_i \le \beta_1^{-1}\hat{Y}_i\,  \hat{Y}_i^\dagger$. These two properties together with~\eqref{eq:ortho} \emph{almost} imply that $\sum_{i\neq j} \Tr_{\rho_i}\big( \hat{Y}_i^\dagger P_i\, P_j\, P_i \hat{Y}_i\big) \lesssim \beta^{-1} \alpha \Tr(\rho)$. Choosing $\beta_1 \approx \sqrt{\alpha}$, we could then apply Lemma~\ref{lem:kproj} to the $P_i$ and states $\sigma_i := \hat{Y}_i \rho_i \hat{Y}_i^\dagger$, recovering close orthogonal projectors $\Pi_i$ which would satisfy the required condition. Carrying out this intuition precisely is a bit tedious, and we now proceed to the details. We will use the following simple fact.

\begin{fact}\label{fact:posblocks} Let $A\geq 0$, $\rho \geq 0$, and $\Pi$ a projection. Let
$$a = \Tr_\rho(A),\qquad b = |\Trho((\Id-\Pi)A\Pi)|\qquad\text{and}\qquad c = \Trho\big((\Id-\Pi)A(\Id-\Pi)\big)$$
Then both the following hold
\begin{align*}
&\Trho\big(\Pi A \Pi\big) \leq \big( \sqrt{a} + \sqrt{c} \big)^2 \leq 2(a+c)\\
& \Trho\big(\Pi A \Pi\big) \,\leq\, \Big(\frac{\sqrt{a} + \sqrt{a+4 b}}{2}\Big)^2 \,\leq\, a + 2b
\end{align*}
\end{fact}

\begin{proof}
Write $\Pi = (\Pi-\Id)+\Id$, so $\Trho(\Pi A \Pi) \leq |\Trho((\Pi-\Id)A\Pi)| + |\Trho( A \Pi)|$. The second term can be bounded by $a^{1/2} \Trho(\Pi A \Pi)^{1/2}$ by Cauchy-Schwarz. 
Similarly bounding the first term by $c^{1/2} \Trho(\Pi A \Pi)^{1/2}$ yields the first equation. To get the second,
let $X = \Trho(\Pi A\Pi)^{1/2}$ to obtain the equation 
$$ X^2 - a^{1/2}X - b \leq 0$$
Solving and using $X \geq 0$, one finds that this is equivalent to $X \leq (\sqrt{a} + \sqrt{a+4 b})/2$.
\end{proof}
 
 Let $Y_{-i}:= \sum_{j\neq i} \hat{Y}_j \hat{Y}_j^\dagger \leq \Id$, and $Q_i$ be the projector on the eigenvectors of $P_i Y_{-i} P_i$ with eigenvalue at most $\beta_2$. Note that, by definition, $Q_i \leq P_i \leq \beta_1^{-1} \hat{Y}_i \hat{Y}_i^\dagger$ (and in particular $Q_i$ commutes with $P_i$). We first bound the distance between $\hat{Y}_i^\dagger$ and $\hat{Y}_i^\dagger Q_i$: since $\hat{Y}_i^\dagger(\Id-Q_i) = \hat{Y}_i^\dagger(\Id-P_i) + \hat{Y}_i^\dagger P_i (\Id-Q_i)P_i$,
\beq\label{eq:ybound0}
\sum_i \Tr_{\rho_i}\big(\hat{Y}_i^\dagger (\Id-Q_i) \hat{Y}_i\big) = \sum_i \Big(\Tr_{\rho_i}\big(\hat{Y}_i^\dagger (Id-P_i) \hat{Y}_i\big) + \Tr_{\rho_i}(\hat{Y}_i^\dagger P_i (Id-Q_i) P_i \hat{Y}_i)\Big)
\eeq
The first term is easily bounded by $\beta_1\, \Tr(\rho)$. For the second, note that $P_i (Id- Q_i) P_i \leq \beta_2^{-1} P_i Y_{-i} P_i$. Using Fact~\ref{fact:posblocks} with $A^i = Y_{-i}$, $\Pi^i = P_i$, and $\rho^i = \hat{Y}_i \rho_i (\hat{Y}_i)^\dagger$ we get $\sum_i a^i \leq \alpha\Tr(\rho)$ and $\sum_i c^i\leq \beta_1\Tr(\rho)$, so that   
\begin{align*}
\sum_i \Tr_{\rho_i}\big( \hat{Y}_i^\dagger P_i Y_{-i} P_i \hat{Y}_i \big) &\leq 2(\alpha + \beta_1)\Tr(\rho)
\end{align*}
Assuming $\alpha \leq \beta_1$ (which will hold for our choice of parameters), from~\eqref{eq:ybound0} we get
\beq\label{eq:ybound}
\sum_i \Tr_{\rho_i}\big(\hat{Y}_i^\dagger (\Id-Q_i) \hat{Y}_i\big) \,\leq\, O(\beta_2^{-1}\beta_1) \Tr(\rho)
\eeq

Next observe that, by definition of $Q_i$, followed by an application of the Cauchy-Schwarz inequality,
\begin{align}
\sum_{i}\big|\Tr_{\rho_i}\big(\hat{Y}_i^\dagger Q_i Y_{-i} (\Id-Q_i) \hat{Y}_i \big)\big| & =\sum_{i}\big|\Tr_{\rho_i}\big(\hat{Y}_i^\dagger Q_i Y_{-i} (\Id-P_i) \hat{Y}_i \big)\big| \notag\\
& \leq \Big( \sum_i\Tr_{\rho_i}\big(\hat{Y}_i^\dagger (\Id-P_i) \hat{Y}_i \big)\Big)^{1/2} \Big( \sum_i \Tr_{\rho_i}(\hat{Y}_i^\dagger Q_i Y_{-i}^2 Q_i \hat{Y}_i \big)\Big)^{1/2}\notag\\
&\leq \beta_1^{1/2} \beta_2 \Tr(\rho)\label{eq:ybound00}
\end{align}
where we used $Q_i Y_{-i}^2 Q_i \leq \beta_2^2 \Id$, which holds by definition of $Q_i$, to bound the second term in the last inequality. Using the second bound in Fact~\ref{fact:posblocks} with $A^i = Y_{-i}$, $\Pi^i = Q_i$, $\rho^i = \hat{Y}_i \rho_i (\hat{Y}_i)^\dagger$, we get $\sum_i a^i\leq\alpha\Tr(\rho)$ and $\sum_i b^i\leq \beta_1^{1/2} \beta_2 \Tr(\rho)$ by~\eqref{eq:ybound00}, so that 
\begin{align*}
\sum_{i\neq j} \Tr_{\rho_i}\big(\hat{Y}_i^\dagger Q_i Q_j Q_i \hat{Y}_i \big) &\leq \beta_1^{-1}\sum_i \Tr_{\rho_i}\big(\hat{Y}_i^\dagger Q_iY_{-i} Q_i\hat{Y}_i \big)  \\
&\leq \beta_1^{-1} \big(\alpha + 2 \beta_1^{1/2} \beta_2\big)\Tr(\rho)
\end{align*}
 Set $\beta_2 = \beta_1^{3/4}$ and $\beta_1  = \alpha^{4/5}$ to obtain 
\beq\label{eq:ybound2}
\sum_{i\neq j} \Tr_{\rho_i}\big(\hat{Y}_i^\dagger Q_i Q_j Q_i \hat{Y}_i \big) \leq O(\alpha^{1/5})\,\Tr(\rho)
\eeq
Let $\sigma_i := \hat{Y}_i\rho_i \hat{Y}_i^\dagger$. We are now ready to apply Lemma~\ref{lem:kproj} to the $Q_i$ and $\sigma_i$: the first condition holds by~\eqref{eq:ybound2}, and the second is a direct consequence of~\eqref{eq:ortho} and $Q_j \leq \beta_1^{-1} \hat{Y}_j\hat{Y}_j^\dagger$ for every $j$. The lemma then gives us pairwise orthogonal $\Pi_i$ such that 
$$\sum_i \Tr_{\rho_i}\big(\hat{Y}_i^\dagger (Q_i-\Pi_i)^2 \hat{Y}_i \big) \leq O(\alpha^{1/10})\Tr(\rho)$$
Combined with~\eqref{eq:ybound} and the triangle inequality, this leads to
$$\sum_i \Tr_{\rho_i}\big(\hat{Y}_i^\dagger (\Id-\Pi_i)\hat{Y}_i \big) \leq O(\alpha^{1/10})\,\Tr(\rho)$$
\end{proof}

\section{Discussion and open questions}

Our work shows for the first time that the entangled value of games can be decreased through parallel repetition. 
Even though we framed and proved our results in the context of $2$-player games, it should not be hard to extend them in some cases to multiple players, depending on the kind of projection or consistency constraints that one can assume on the game. On the other hand, extending the result to either many-round games, or games with quantum messages, is an interesting open question. 

One implication of our result is the following. The celebrated
PCP theorem says that given a game, it is NP-hard to tell if its value is $1$ or less than, say, $0.99$.
Combined with Raz's parallel repetition result, one obtains that it is also hard to tell if the value is $1$
or less than, say, $0.01$. The latter statement led to an enormous body of work on strong hardness of approximation
results~\cite{Hastad01}. It is currently a major open question whether an analogue of the PCP theorem
holds for the entangled value. \emph{If} such a result was proved, our results would
allow to amplify the hardness to $1$ vs.\ $0.01$, as in the classical case, possibly leading to further
surprising implications.

The main open question left by our work is whether it is possible to show a better rate of decay, in particular an
exponential rate as Raz obtained from direct parallel repetition, or~\cite{Impagliazzo2009} first obtained in the
setting of direct product testers. Another open question is whether our statement can be extended to hold for simple parallel repetition for arbitrary entangled games (i.e. without adding dummy or consistency questions).

We believe that our main conceptual contributions are the extension of the notion of ``approximately serial" to the
setting of measurements, and our subsequent {\em orthogonalization lemma}. We hope that these techniques might prove
useful elsewhere, perhaps in establishing hardness of entangled games. Lastly, product testers are very useful in
the area of property testing, and it remains to be seen if our result can be applied similarly.

\paragraph{Acknowledgments.} We are indebted to Ryan O'Donnell for making publicly available his extremely clear and helpful lecture notes~\cite{O'Donnell2005,O'Donnell2005a} on Feige and Kilian's parallel repetition result,  and to user ``ohai'' of MathOverflow.net for pointing out the connection between the classical Procrustes problem and that of the robust orthonormalization of almost-orthogonal families of vectors. We especially thank Oded Regev for useful discussions and helpful comments, Tsuyoshi Ito and Ben Reichardt for comments, and Ben Reichard for pointing out an error in the proof of Claim~\ref{claim:exptrace2} in a previous version of this manuscript. 

\bibliographystyle{alphaabbrv}
\bibliography{par-rep3}

\appendix


\section{Some useful technical facts}\label{sec:exp}

In this section we prove a series of useful claims showing that, in a strategy which has been marginalized over a large number of indices, fixing a particular coordinate $(i,q_i)$ does not have much influence on average. Throughout this question we fix a question set $Q$ and a distribution $\mu$ on $Q$. Whenever an expectation over tuples of questions $q\in Q^C$ is taken, it will be over the product distribution $\mu^C$.

Our claims will rely essentially on the following, which applies to \emph{any} matrix semi-norm $\|\cdot\|$, provided it is derived from a semi-inner product $\langle \cdot,\cdot\rangle$.

\begin{claim}\label{claim:exp}
Let $C$ be an integer, and $f:Q^C \rightarrow \{\, X\in\C^{d\times d}\,\}$. Let $M = \Exs{q}{f(q)}$ and for any $(i,q_i)$, $M_{i,q_i} = \Exs{q_{\neg i}}{f(q)}$. Suppose that $\Exs{q}{\|f(q)\|^2} \leq 1$. Then
\begin{enumerate}
\item $0 \leq \Exs{i,q_i}{\|M-M_{i,q_i}\|^2} \leq \frac{\Exs{q}{\|f(q)\|^2} }{C} \leq \frac{1}{C}$. \item
$\Exs{i,q_i}{\|M-M_{i,q_i}\|^2}  = \Exs{i,q_i}{\|M_{i,q_i}\|^2} - \|M\|^2$. \item
$\Pr_{i,q_i}(|\Tr(M)-\Tr(M_{i,q_i})|\geq C^{-1/3}  ) \leq C^{-1/3}$.
\end{enumerate}
\end{claim}

\begin{proof} The proof of all three parts is in close analogy to that of Lemma~2.1 in~\cite{O'Donnell2005a}, which shows
similar statements for a {\em Boolean} function $f$. For part 1 note that $\Exs{i,q_i}{\|M-M_{i,q_i}\|^2}=\frac{1}{C}
\sum_{i=1}^C \Exs{q_i}{\|M-M_{i,q_i}\|^2}$ and hence it suffices to show that $\sum_{i=1}^C
\Exs{q_i}{\|M-M_{i,q_i}\|^2}\leq \Tr(M)$. Observe that
\begin{align*}
0 & \leq \Exs{q}{\|f(q)-\sum_i (M_{i,q_i}-M)\|^2}\\
& =\Exs{q}{\|f(q)\|^2}- \sum_i \Exs{q_i}{\langle M_{i,q_i}-M, M_{i,q_i}\rangle+\langle M_{i,q_i}, M_{i,q_i}-M\rangle}+\sum_{i,j}\Exs{q_i,q_j}{\langle M-M_{i,q_i},M-M_{j,q_j}\rangle}\\
& = \Exs{q}{\|f(q)\|^2}-\sum_i\Exs{q_i}{\|M-M_{i,q_i}\|^2},
\end{align*}
where for the last equality we have used that $\Exs{q_i}{M_{i,q_i}-M}=0$ and hence $\Exs{q_i}{\langle M_{i,q_i}-M,
M_{i,q_i}\rangle}=\Exs{q_i}{\langle M_{i,q_i}-M, M_{i,q_i}-M\rangle}$ and, for $i \neq j$, 
$$\Exs{q_i,q_j}{\langle M-M_{i,q_i},M-M_{j,q_j}\rangle}=\langle \Exs{q_i}{M-M_{i,q_i}}, \Exs{q_j}{M-M_{j,q_j}}\rangle=0$$
Part 1. now follows, and the second inequality is simply the assumption that $\Exs{q}{\|f(q)\|^2} \leq 1$.

Part 2 is trivial from the expansion of $\|M-M_{i,q_i}\|^2$. Part 3 follows from part 1 using Markov's inequality,
which gives $\Pr_{i,q_i}((\Tr(M-M_{i,q_i}))^2 \geq C^{-2/3}) \leq C^{2/3} \Exs{i,q_i}{(\Tr(M-M_{i,q_i}))^2}$. Observing
that for $A:=M-M_{i,q_i}$ we have $(\Tr(A))^2=\langle A,\Id \rangle^2 \leq \|A\|^2 \cdot \|\Id\|^2=\|A\|^2$ gives the
desired bound.
\end{proof}

The following is a direct corollary of Claim~\ref{claim:exp}, obtained for a specific instantiation of the norm $\|\cdot\|$.

\begin{claim}\label{claim:exptrace3}
Let $Y_q^a$, for $q\in Q^C$ and $a\in A^C$, be positive matrices such that $Y_q:= \sum_a Y_q^a \leq \Id$, and $\rho\geq 0$. Let $Y = \Exs{q}{Y_q}$. Then 
$$ \Exs{(i,q_i)}{\big|\Tr\big(Y\,\rho^{1/2} Y\,\rho^{1/2} \big) -  \Tr\big(Y_{q_i }\,\rho^{1/2}Y_{q_i}\,\rho^{1/2} \big)\big|} \,\leq\, C^{-1} \Exs{q}{\Tr\big(Y_q\rho^{1/2} Y_q \rho^{1/2}\big)}\, \leq\, \Trho(Y) $$
\end{claim}

\begin{proof} The statement follows from Claim~\ref{claim:exp}, applied to $f(q) = Y_q$ and the (semi)-norm $\|A\|^2 = \Tr\big( A \rho^{1/2} A^\dagger \rho^{1/2}\big)$, which is derived from the inner-product $(A,B) \mapsto \Tr\big(A \rho^{1/2} B^\dagger \rho^{1/2}\big)$. The second inequality holds since $0\leq Y_q \leq \Id$ for every $q$. 
\end{proof}

We now give two simple calculations which will be useful. The first is a well-known operator version of the Cauchy-Schwarz inequality.

\begin{claim}\label{claim:matrixcs} Let $A,B$ be (possibly rectangular) matrices such that $A^\dagger B$ exists, and $B^\dagger B$ is invertible. Then
$$ (A^\dagger B)(B^\dagger B)^{-1} (B^\dagger A) \,\leq\, A^\dagger A$$
\end{claim}

\begin{proof} Let $\Delta = (B^\dagger B)^{-1} (B^\dagger A)$. Then the matrix $(A-B\Delta)^\dagger (A-B\Delta)$ is positive, which gives the result. 
\end{proof}

\begin{claim}\label{claim:expsquare}
Let $Y_q \in\C^{d\times d}$, $0\leq Y_q\leq\Id$, for $q\in Q^C$, and let $Y = \Exs{q}{Y_q}$, $Y_{i,q_i} = \Exs{q_{\neg i}}{Y_q}$ for $i\in [C]$. Then
$$\Exs{(i,q_i)}{ (Y-Y_{i,q_i})^2} \,\leq C^{-1}\Exs{q}{Y_q^2}$$
\end{claim}

\begin{proof}
Write
\begin{align*}
 0 &\leq \Big(Y_q - \sum_i( Y_{i,q_i}-Y)\Big)\Big(Y_q - \sum_i( Y_{i,q_i}-Y)\Big) \\
 & = Y_q^2 - \sum_i \big( Y_q( Y_{i,q_i}-Y) + \big(Y_{i,q_i}-Y)Y_q  \big)+ \sum_{i,j} \big( Y_{i,q_i}-Y\big)\big( Y_{j,q_j}-Y\big)
\end{align*}
Taking the expectation over $q$, we obtain
$$
\sum_{i} \Exs{q_i}{( Y_{i,q_i}-Y)^2}\,\leq\, \textsc{E}_q\big[Y_q^2\big]
$$
Dividing by $C$ on both sides proves the claim. 
\end{proof}

%
%

\begin{claim}\label{claim:exptrace2}
For every $q\in Q^C$ let $\{X_q^a\}_{a\in A^{C'}}$ be a POVM, and $\hat{X}_q^a := \sqrt{\pi(q)} \sqrt{X_q^a} \otimes \bra{q,a}$ (as described in Section~\ref{sec:notation}), and $\rho\geq 0$. Assume that $\hat{X} \hat{X}^\dagger = \sum_a\Exs{q}{\hat{X}_q^a (\hat{X}_q^a)^\dagger} \leq \Id$. Then 
$$ \sum_{a} \Exs{(i,q_i)}{ \big| \Trho\big( (\hat{X}^a)^\dagger \hat{X}^a(\hat{X}^a)^\dagger \hat{X}^a\big) - \Trho\big( (\hat{X}_{q_i}^a)^\dagger \hat{X}_{q_i}^a(\hat{X}_{q_i}^a)^\dagger \hat{X}_{q_i}^a\big) \big|} \,\leq\, 2\,C^{-1/2} \Tr(\rho)$$
\end{claim}

\begin{proof}
Let $\tilde{X}_i^a = \big|\hat{X}^a(\hat{X}^a)^\dagger - \hat{X}_{q_i}^a(\hat{X}_{q_i}^a)^\dagger\big|$, and $\tilde{\rho}_i^a = \big|\hat{X}^a\rho(\hat{X}^a)^\dagger-\hat{X}_{q_i}^a\rho(\hat{X}_{q_i}^a)^\dagger\big|$, where the notation keeps the dependence on $q_i$ implicit. Use the triangle inequality to write
\begin{align}
\big|\Tr\big( \hat{X}^a (\hat{X}^a)^\dagger \hat{X}^a\rho (\hat{X}^a)^\dagger\big) -  \Tr\big( \hat{X}_{q_i}^a (\hat{X}_{q_i}^a)^\dagger \hat{X}_{q_i}^a\rho (\hat{X}_{q_i}^a)^\dagger\big)\big| &\leq \Tr\big( \tilde{X}_i^a \hat{X}^a\rho (\hat{X}^a)^\dagger\big) + \Tr\big(\hat{X}_{q_i}^a (\hat{X}_{q_i}^a)^\dagger \tilde{\rho}_i^a\big)\label{eq:exptrace5}
\end{align}
The expectation of the first term on the right-hand side of~\eqref{eq:exptrace5} can be bounded by Cauchy-Schwarz as
\begin{align*}
\Exs{(i,q_i)}{\Tr\big( \tilde{X}_i^a \hat{X}^a\rho (\hat{X}^a)^\dagger\big)} &\leq \Exs{(i,q_i)}{\Tr_\rho((\hat{X}^a)^\dagger \hat{X}^a)^{1/2}  \Tr\big( (\tilde{X}^a_i)^2 \hat{X}^a\rho (\hat{X}^a)^\dagger\big)^{1/2}} \\
&\leq C^{-1/2} \Tr_\rho((\hat{X}^a)^\dagger \hat{X}^a)
\end{align*}
by Claim~\ref{claim:exp}, applied to the (semi)-norm $\|A\|^2 := \Tr\big((A^\dagger A)\,(\hat{X}^a\rho (\hat{X}^a)^\dagger)\big)$ and the mapping $f:q\mapsto \hat{X}^a_q (\hat{X}_q^a)^\dagger$.

Regarding the second term on the right-hand side of~\eqref{eq:exptrace5}, let $A$ be the block-column matrix with blocks $\sqrt{\pi(q_i)}\tilde{\rho}_i^a$ for every $(i,q_i)$ and $a$, and $B$ with blocks $\sqrt{\pi(q_i)}\hat{X}_i^a(\hat{X}_i^a)^\dagger$. Then $B^\dagger B = \sum_a \Exs{(i,q_i)}{\big(\hat{X}_i^a (\hat{X}_i^a)^\dagger\big)^2} \leq \Id$. Let $D = A^\dagger B = \sum_a \Exs{(i,q_i)}{\tilde{\rho}_i^a \hat{X}_i^a(\hat{X}_i^a)^\dagger}$; the operator Cauchy-Schwarz inequality from Claim~\ref{claim:matrixcs} gives
$$
DD^\dagger \,\leq\, D(B^\dagger B)^{-1}D^\dagger\,\leq\, A^\dagger A \,=\, \sum_a \Exs{(i,q_i)}{(\tilde{\rho}_i^a)^2 }
$$
Applying Claim~\ref{claim:expsquare} to $\hat{X}_{q}^a\rho(\hat{X}_{q}^a)^\dagger$ (for every $a$), we can then bound
\beq\label{eq:dd}
 DD^\dagger \,\leq\,C^{-1} \Exs{q}{ (\hat{X}_q \rho \hat{X}_q^\dagger)^2}\,\leq\, C^{-1} \Exs{q}{ \hat{X}_q \rho^2 \hat{X}_q^\dagger} 
 \eeq
 where for the second inequality we used $ \hat{X}_q^\dagger \hat{X}_q\leq \Id$. 
Since $\Tr(D) \leq \Tr\big(\sqrt{DD^\dagger}\big) = \|D\|_1$, taking the square root on both sides of~\eqref{eq:dd} (the square root being operator monotone) and then the trace, we obtain
$$
\sum_a \Exs{(i,q_i)}{\Tr\big(\tilde{\rho}_i^a \hat{X}_i^a(\hat{X}_i^a)^\dagger \big)} \,\leq\, C^{-1/2} \Tr\sqrt{ \Exs{q}{ \hat{X}_q \rho^2 \hat{X}_q^\dagger}}\,=\, C^{-1/2} \big\| \hat{X} \rho \big\|_1
$$
where $\hat{X}$ is the rectangular matrix with square blocks $\pi(q)^{-1/2}\hat{X}_q^a$ arranged in a column. By Holder's inequality $\big\| \hat{X} \rho \big\|_1 \leq \Tr(\rho) \|\hat{X}\|_\infty$, and $\|\hat{X}\|_\infty \leq 1$ since $\hat{X}^\dagger\hat{X} = \Exs{q}{ \hat{X}_q^\dagger \hat{X}_q} \leq \Id$. This finishes the proof of the claim. 
\end{proof}

\end{document}